\documentclass[onecolumn,secnumarabic,nobibnotes,aps]{revtex4-2}
\usepackage{color, xcolor, colortbl}
\usepackage{graphicx}
\usepackage{geometry}
\usepackage{amsthm,amsmath,amssymb,amsfonts}
\usepackage{algorithm}
\usepackage{algorithmic}
\usepackage{dcolumn}
\usepackage{epstopdf}
\usepackage{bm}
\usepackage[caption=false]{subfig}
\usepackage{comment}
\usepackage{appendix}
\usepackage{multirow}
\usepackage{braket}
\usepackage[english]{babel}
\usepackage{hyperref}
\usepackage[capitalize]{cleveref}
\usepackage{xpatch}
\usepackage{tikz}
\usepackage{adjustbox}
\usepackage{xspace}
\usepackage{complexity}

\renewcommand{\Re}{\operatorname{Re}}

\renewcommand{\polylog}{\operatorname{poly}\log}
\renewcommand{\poly}{\operatorname{poly}}

\DeclareMathOperator*{\argmin}{arg\,min}

\newcommand{\I}{i}

\newcommand{\opr}[1]{\ensuremath{\operatorname{#1}}}
\newcommand{\mc}[1]{\mathcal{#1}}

\newcommand{\wt}[1]{\widetilde{#1}}

\newcommand{\abs}[1]{\left\lvert#1\right\rvert}
\newcommand{\norm}[1]{\left\lVert#1\right\rVert}

\newcommand{\Or}{\mathcal{O}}

\newcommand{\RR}{\mathbb{R}}

\newtheorem{thm}{\protect\theoremname}
\newtheorem{lem}[thm]{\protect\lemmaname}
\newtheorem*{lem*}{\protect\lemmaname}

\newtheorem{cor}[thm]{\protect\corollaryname}

\newtheorem{defn}[thm]{\protect\definitionname}
\newtheorem*{defn*}{\protect\definitionname}

\providecommand{\definitionname}{Definition}
\providecommand{\assumptionname}{Assumption}
\providecommand{\corollaryname}{Corollary}
\providecommand{\lemmaname}{Lemma}
\providecommand{\propositionname}{Proposition}
\providecommand{\remarkname}{Remark}
\providecommand{\theoremname}{Theorem}

\usetikzlibrary{quantikz}

\usetikzlibrary{fit}
\tikzset{%
  highlight/.style={rectangle,rounded corners,fill=blue!15,draw,fill opacity=0.3,thick,inner sep=0pt}
}

\usepackage{cellspace} 
\usepackage{makecell} 
\setcellgapes{3pt}

\newcommand{\QETU}{{QET-U}\xspace}
\newcommand{\TFIM}{\ensuremath{\mathrm{TFIM}}}

\newcommand{\red}[1]{#1}
\newcommand{\REV}[1]{#1}
\newcommand{\COR}[1]{#1}

\usepackage{dsfont}
\newcommand{\bP}{\mathds{P}}

\begin{document}

\newcommand{\DeptMath}{Department of Mathematics, University of California, Berkeley, California 94720 USA}
\newcommand{\LBLMath}{Applied Mathematics and Computational Research Division, Lawrence Berkeley National Laboratory, Berkeley, CA 94720, USA}
\newcommand{\BQIC}{Berkeley Center for Quantum Information and Computation, Berkeley, California 94720 USA}
\newcommand{\CIQC}{Challenge Institute of Quantum Computation, University of California, Berkeley, California 94720 USA}

\title{Ground-state preparation and energy estimation \\on early fault-tolerant quantum computers \\via quantum eigenvalue transformation of unitary matrices}
\author{Yulong Dong$^{1,4}$} 
\author{Lin Lin$^{1,2,3}$}
\email{Electronic address: linlin@math.berkeley.edu}
\author{Yu Tong$^{1}$} 
\affiliation{$^1$\DeptMath}
\affiliation{$^2$\LBLMath}
\affiliation{$^3$\CIQC}
\affiliation{$^4$\BQIC}

\date{\today}

\begin{abstract}
Under suitable assumptions, the algorithms in [Lin, Tong, Quantum 2020] can estimate the ground-state energy and prepare the ground state of a quantum Hamiltonian with near-optimal query complexities. However, this is based on a block encoding input model of the Hamiltonian, whose implementation is known to require a large resource overhead. We develop a tool called quantum eigenvalue transformation of unitary matrices with real polynomials (QET-U), which uses a controlled Hamiltonian evolution as the input model, a single ancilla qubit and no multi-qubit control operations, and is thus suitable for early fault-tolerant quantum devices. This leads to a simple quantum algorithm that outperforms all previous algorithms with  a comparable circuit structure for estimating the ground-state energy. 
For a class of quantum spin Hamiltonians, we propose a new method that exploits certain anti-commutation relations and further removes the need of implementing the controlled Hamiltonian evolution. Coupled with a Trotter-based approximation of the Hamiltonian evolution, the resulting algorithm can be very suitable for early fault-tolerant quantum devices.
We demonstrate the performance of the algorithm using \textsf{IBM Qiskit} for the transverse field Ising model. 
If we are further allowed to use multi-qubit Toffoli gates, we can then implement amplitude amplification and a new binary amplitude estimation algorithm, which increases the circuit depth but decreases the total query complexity. The resulting algorithm saturates the near-optimal complexity for ground-state preparation and energy estimating using  a constant number of ancilla qubits (no more than $3$).
\end{abstract}

\maketitle

\section{Introduction}

Preparing the ground state and estimating the ground-state energy of 
a quantum Hamiltonian have a wide range of applications in condensed matter physics, quantum chemistry, and quantum information. 
To solve such problems, quantum computers promise to deliver a new level of computational power that can be significantly  beyond the boundaries set by classical computers.
Despite exciting early progress on NISQ devices~\cite{Preskill2018}, it is widely believed that most scientific advances in quantum sciences require some version of fault-tolerant quantum computers, which are expected to be able to accomplish much more complicated tasks. 
On the other hand, the fabrication of full-scale fault-tolerant quantum computers remains a formidable technical challenge for the foreseeable future, and it is reasonable to expect that early fault-tolerant quantum computers share the following characteristics: (1) The number of logical qubits is limited. (2) It can be difficult to execute certain controlled operations (e.g., multi-qubit control gates), whose implementation require a large number of non-Clifford gates. Besides these, the maximum circuit depth of early-fault-tolerant quantum computers, which is determined by the maximum coherence time of the devices, may still be limited. Therefore it is still important to reduce the circuit depth, sometimes even at the expense of a larger total runtime (via a larger number of repetitions). 
Quantum algorithms tailored for early fault-tolerant quantum computers \red{\cite{campbell2020early,BabbushMcCleanEtAl2021focus,BoothOGorman2021quantum,Layden2021first,LinTong2022,WanBertaCampbell2021randomized,ZhangWangJohnson2021computing,WangSimJohnson2022state}} need to properly take these limitations into account, and the resulting algorithmic structure can be different from those designed for fully fault-tolerant quantum computers.

To gain access to the quantum Hamiltonian $H$, a standard input model is the \textit{block encoding} (BE) model, which directly encodes the matrix $H$ (after proper rescaling) as a submatrix block of a larger unitary matrix $U_H$~\cite{LowChuang2019,ChakrabortyGilyenJeffery2018}. Combined with techniques such as linear combination of unitaries (LCU)~\cite{BerryChildsKothari2015}, quantum signal processing~\cite{LowChuang2017} or quantum singular value transformation~\cite{GilyenSuLowEtAl2019}, one can implement a large class of matrix functions of $H$ on a quantum computer. This leads to quantum algorithms for ground-state preparation and ground-state energy estimation with near-optimal query complexities to $U_H$~\cite{LinTong2020a}. The block encoding technique is also very useful in many other tasks such as Hamiltonian simulation, solving linear systems, preparing the Gibbs state, and computing Green's function and the correlation functions \cite{TongAnWiebeEtAl2020,ChakrabortyGilyenJeffery2018,Rall2020quantum,GilyenSuLowEtAl2018,LowChuang2017}. However, the block encoding of a quantum Hamiltonian (e.g., a sparse matrix) often involves a relatively large number of ancilla qubits, as well as multi-qubit controlled operations that lead to a large number of two-qubit gates and long circuit depths~\cite{GilyenSuLowEtAl2019}, and is therefore not suitable in the early fault-tolerant setting. 

A widely used alternative approach for accessing the information in $H$ is the time evolution operator $U=\exp(-i \tau H)$ for some time $\tau$. 
This input model will be referred to as the \textit{Hamiltonian evolution} (HE) model. While Hamiltonian simulation can be performed using quantum signal processing for sparse Hamiltonians with optimal query complexity~\cite{LowChuang2017}, such an algorithm queries a block encoding of $H$, which defeats the purpose of employing the HE model. 
On the other hand, when $H$ can be efficiently decomposed into a linear combination of Pauli operators, the time evolution operator can be efficiently implemented using, e.g., the Trotter product formula~\cite{Lloyd1996,ChildsSuTranEtAl2021} without using any ancilla qubit. 
This remarkable feature has inspired quantum algorithms for performing a variety of tasks using controlled time evolution and one ancilla qubit. 
A textbook example of such an algorithm is the Hadamard test. It uses one ancilla qubit and the controlled Hamiltonian evolution to estimate the average value $\Re\braket{\psi|U|\psi}$, which is encoded by the probability of measuring the ancilla qubit with outcome $0$ (see \cref{fig:main_circuits} (a)). 
The number of repeated measurements of this procedure is $\Or(\epsilon^{-2})$, where $\epsilon$ is the desired precision. 
Assume the spectrum of the Hamiltonian $H$ is contained in $[\eta,\pi-\eta]$ for some $\eta>0$.
If $\ket{\psi}$ is the exact ground state of $H$, we can retrieve the eigenvalue as $\lambda=\arccos(\Re\braket{\psi|U|\psi})$.  
By allowing a series of longer simulation times of $t=d$ for some integer $d$, this leads to Kitaev's algorithm that uses only $\log \epsilon^{-1}$ measurements, at the expense of increasing the circuit depth to $\Or(\epsilon^{-1})$ (see \cref{fig:main_circuits} (b)).
The total simulation time is therefore $\Or(\epsilon^{-1}\log\epsilon^{-1})$, which reaches the Heisenberg limit~\cite{GiovannettiLloydMaccone2006,GiovannettiLloydMaccone2011advances,ZwierzPerezDelgadoKok2010,ZwierzPerezDelgadoKok2012ultimate} up to a logarithmic factor.

When the input state $\ket{\phi_0}$ (prepared by an oracle $U_I$) is different from the exact ground state of $H$ denoted by $\ket{\psi_0}$,  there have been multiple quantum algorithms using the circuit of \cref{fig:main_circuits} (b) or its variant to estimate the ground-state energy~\cite{wang2019accelerated,wiebe2015bayesian,OBrienTarasinskiTerhal2019quantum,LinTong2022}.
Let $\gamma$ be a lower bound of the initial overlap, i.e., $|\braket{\phi_0|\psi_0}|\ge \gamma$. It is worth noting that \textit{all} quantum algorithms with provable performance guarantees require \textit{a priori} knowledge that $\gamma$ is reasonably large (assuming black-box access to the Hamiltonian). Without such an assumption, this problem is \QMA-hard \cite{KitaevShenVyalyi2002, KempeKitaevRegev2006, OliveiraTerhal2005, AharonovGottesmanEtAl2009}.
Candidates for such $\ket{\phi_0}$ include the Hartree-Fock state in quantum chemistry \cite{KivlichanMcCleanEtAl2018,TubmanEtAl2018postponing}, and quantum states prepared using the variational quantum eigensolver \cite{PeruzzoMcCleanShadboltEtAl2014,McCleanRomeroBabbushEtAl2016,omalley2016scalable}. \red{Some techniques can be used to boost the overlap using low-depth circuits \cite{WangSimJohnson2022state}.}
Furthermore, algorithms using the circuit of \cref{fig:main_circuits} (b) typically cannot be used to prepare the ground state. \red{With the time evolution operator as the input, one can use the LCU algorithm to prepare the ground state \cite{ge2019faster,KeenDumitrescuWang2021quantum}, thus reducing the number of ancilla qubits needed to implement the block encoding of the Hamiltonian. Note that LCU requires additional ancilla qubits to store the coefficients, and as a result cannot be implemented using $\Or(1)$ qubits.}

\begin{figure}[H]
\begin{center}
\begin{center}
\subfloat[Hadamard test, short time evolution]{
\includegraphics[width=0.376\linewidth]{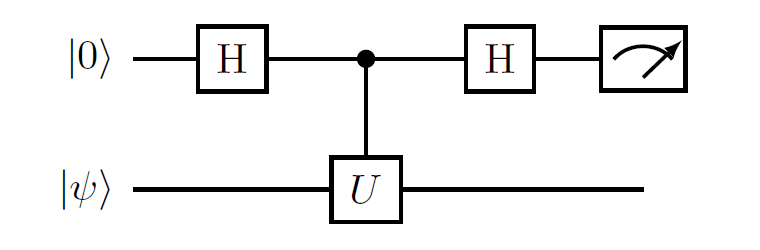}
}
\hspace{2em}
\subfloat[Hadamard test, long time evolution]{
\includegraphics[width=0.424\linewidth]{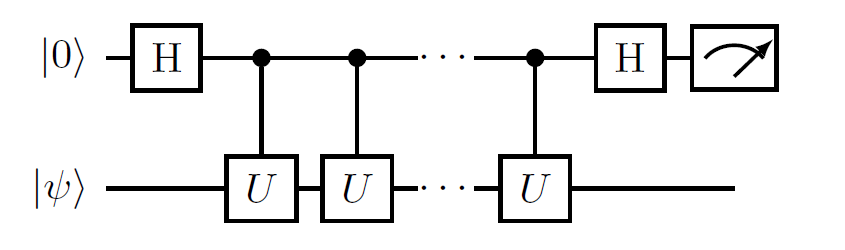}
}

\subfloat[\QETU]{
  \includegraphics[width=0.75\linewidth]{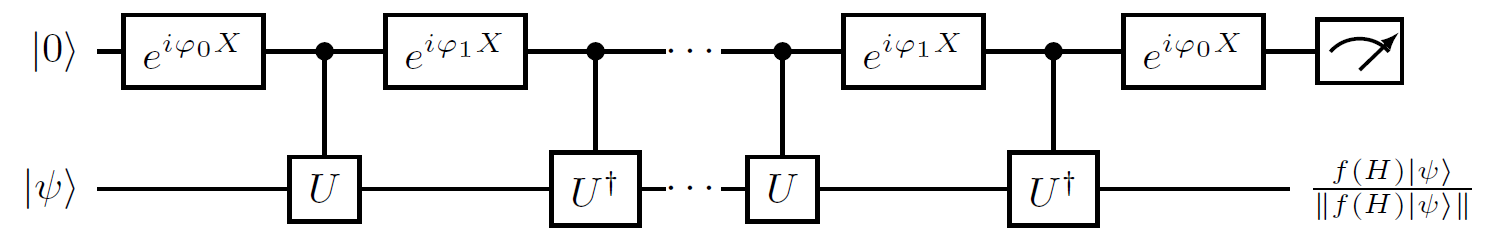}
}
\end{center}
\end{center}
\caption{After a proper rescaling, the $n$-qubit circuit $U$ implements $e^{-i H}$, and Hadamard test circuit (a) estimates $\Re\braket{\psi|e^{-i H}|\psi}$.
Repeating the controlled evolution $d$ times, the circuit (b) estimates the average value of a long time evolution $\Re\braket{\psi|e^{-i d H}|\psi}$.
For a very general class of functions $f$ the circuit (c) can approximately prepare a normalized quantum state $f(H)\ket{\psi}/\norm{f(H)\ket{\psi}}$ with approximate success probability $p=\norm{f(H)\ket{\psi}}^2$, by interleaving the forward ($U$) and backward ($U^{\dag}$)  time evolution with some properly chosen $X$-rotations in the ancilla qubit.}
\label{fig:main_circuits}
\end{figure}

We will show that both the ground-state preparation and energy estimation can be solved by (repeatedly) preparing a quantum state of the form $\ket{\psi_f}\propto f(H)\ket{\phi_0}$, where $f$ is a real polynomial approximating a shifted sign function.
A main technical tool developed in this this paper is called \textit{quantum eigenvalue transformation of unitary matrices with real polynomials} (\QETU), which allows us to prepare such a state $\ket{\psi_f}$ by querying $U=e^{-i  H}$, using only one ancilla qubit, and does not require any multi-qubit control operation (\cref{thm:qet_unitary}). The circuit structure (Fig. \ref{fig:main_circuits} (c)) is only slightly different from that  in Fig \ref{fig:main_circuits} (b).
The \QETU technique is closely related to concepts such as quantum signal processing, quantum eigenvalue transformation, and quantum singular value transformation. The relations among these techniques are detailed in \cref{sec:related_matrixtrans}. The information of the function $f$ of interest is stored in the adjustable parameters $\{\varphi_i\}$ called \textit{phase factors}. To find such parameters, we need to identify  a polynomial approximation to the shifted sign function, and then evaluate the phase factors corresponding to the approximate polynomial. Most quantum signal processing (QSP) based applications construct such a polynomial approximation analytically, which can sometimes lead to cumbersome expressions and suboptimal approximation results. We provide a convex-optimization-based procedure to streamline this process and yield the near-optimal approximation (see \cref{sec:convex}). Both the \QETU technique and the convex optimization method can be useful in applications beyond ground-state preparation and energy estimation.

The computational cost will be primarily measured in terms of the query complexity, i.e., how many times we need to query $U$ and $U_I$ in total, and we will also analyze the additional one- and two-qubit gates needed, such as the single qubit rotation gates, and the two-qubit gates needed to implement the $n$-qubit reflection operator. In our algorithms, the number of additional gates has the same scaling as the query complexity, or involves an $n$ factor where $n$ is system size. We measure the circuit depth requirement in terms of \textit{query depth}: the number of times we need to query $U$ in one coherent run of the circuit. \red{Note that this term is not to be confused with the circuit depth for implementing the oracle $U$, which we will not consider in this work.}  This metric reflects the circuit depth requirement faithfully because in our algorithm, in one coherent run of the circuit, the number of queries to $U_I$ is also upper bounded by this metric, and the additional circuit depth needed for additional gates is upper bounded by this metric up to a factor of $\Or(n)$. 

Besides the query depth, we also focus on whether multi-qubit control needs to be implemented.  Algorithms such as amplitude amplification and amplitude estimation~\cite{BrassardHoyerMoscaEtAl2002}
can be used to reduce the total query complexity, but they also need to use $(n+1)$-bit Toffoli gates (specifically, $n$-qubit reflection operator with respect to the zero state $\ket{0^n}$), which can be implemented using $\Or(n)$ two-qubit gates and one ancilla qubit~\cite{BarencoBennettEtAl1995elementary}. These operations will be referred to as ``low-level'' multi-qubit control gates. Some other quantum algorithms may require more complex multi-qubit control operations as well as more ancilla qubits. For instance, the high-confidence QPE algorithm \cite{knill2007optimal,poulin2009sampling,nagaj2009fast} requires a circuit to carry out the arithmetic operation of taking the median of multiple energy measurement results, which can require $\poly(n)$ two-qubit gates and ancilla qubits. Such operations will be referred to as ``high-level'' multi-qubit control gates.

To solve the ground-state preparation and energy estimation problem, we propose two different types of algorithms, with two different goals in mind. 
For the first type of algorithms, which we call the \textit{short query depth algorithms}, we only use \QETU to prioritize reducing the quantum resources needed. No multi-qubit controlled operation is involved. For the second type of algorithms, which we call the \textit{near-optimal algorithms}, we optimize the total query complexity by using amplitude amplification and a new binary amplitude estimation algorithm (\cref{lem:binary_amplitude_estimation}).
Such algorithms only use low-level multi-qubit control operations. Both types of algorithms only use a small number of ancilla qubits (no more than 2 or 3).

For ground-state energy estimation, surprisingly, even though the total query complexity of the short query depth algorithm does not have the optimal asymptotic scaling, it still outperforms \textit{all} previous algorithms with the same ancilla qubit number constraint \cite{HigginsBerryEtAl2007,BerryHiggins2009,somma2019quantum,LinTong2022}, in terms of total query complexity (see Table \ref{tab:compare_algs_energy}). Most notably, in this setting we achieve a quadratic improvement on the $\gamma$ dependence, from $\wt{\Or}(\gamma^{-4})$ to $\wt{\Or}(\gamma^{-2})$ (the notation $\wt{\Or}(g)$ means $\Or(g\polylog(g))$ unless otherwise stated). Moreover the circuit depth from the previous state-of-the-art result is preserved in our algorithm. Numerical comparison \REV{(see \cref{fig:compare_performance})} demonstrates that our algorithm  outperforms QPE, not only in terms of the asymptotic scaling, but also the exact non-asymptotic number of queries \COR{for moderately small values of $\gamma$}. Our near-optimal algorithm takes this advantage even further, matching the best known query complexity scaling in Ref. \cite{LinTong2020a} (which saturates the query complexity lower bound). 

For ground-state preparation, the only other algorithm that can use at most constantly many ancilla qubits is the quantum phase estimation algorithm with semi-classical Fourier transform \cite{griffiths1996semiclassical}. Compared to this algorithm, our short query depth algorithm has an exponentially improved precision dependence and a quadratically improved $\gamma$ dependence, from $\wt{\Or}(\gamma^{-4})$ to $\wt{\Or}(\gamma^{-2})$, while maintaining the same circuit depth. The near-optimal algorithm further improves the dependence to $\wt{\Or}(\gamma^{-1})$. A comparison of the algorithms for ground-state preparation can be found in Table \ref{tab:compare_algs_state}. We remark that here we consider the case where we know a parameter $\mu$ such that $\lambda_0\leq \mu-\Delta/2<\mu+\Delta/2\leq \lambda_1$, as in Theorem \ref{thm:ground_state_prep_with_bound_wo_AA}. If no such $\mu$ is known, we need to first estimate the ground-state energy to precision $\Or(\Delta)$, and the resulting algorithm is discussed in Theorem \ref{thm:ground_state_prep_full}. \red{If the ground-state energy is known \textit{a priori}, then the algorithm in \cite{ChoiLeeEtAl2021rodeo} may yield a similar speedup for preparing the ground state, but such knowledge  is generally not available.}

In the above analysis, specifically in Tables~\ref{tab:compare_algs_energy} and \ref{tab:compare_algs_state}, we compared with algorithms whose complexity can be rigorously analyzed under the assumptions of a good initial overlap (and spectral gap for ground-state preparation). We did not compare with heuristic algorithms such as the variational quantum eigensolver \cite{PeruzzoMcCleanShadboltEtAl2014,McCleanRomeroBabbushEtAl2016,omalley2016scalable}. There are also algorithms that are designed with different but similar goals in mind, such as the quantum algorithmic cooling technique in Ref.~\cite{ZengSunYuan2021universal}, which can estimate an eigenvalue $\lambda_j$ belonging to a given range $[\lambda_j^L,\lambda_j^R]$ (assuming all other eigenvalues are away from this range). Then the total runtime scaling is $\wt{\Or}(\gamma^{-4})$ where $\gamma$ is the overlap between the initial guess and the target eigenstate \cite[Theorem 2]{ZengSunYuan2021universal}. The same technique can also be used to estimate the observable expectation value of the target eigenstate, without coherently preparing the target eigenstate. 

For certain Hamiltonians, \QETU can be implemented with the standard Hamiltonian evolution rather than the controlled version. Note that ``control-free'' only means that the Hamiltonian evolution is not controlled by one or more qubits, but control gates that are independent of the Hamiltonian can still be used. In Ref. \cite{Huggins2020nonorthogonalVQE}, the control-free setting for an $n$-qubit time evolution is achieved by introducing an $n$-qubit reference state, on which the time evolution acts trivially. The algorithm also requires the implementation of the controlled $n$-qubit SWAP gate, and therefore has a relatively large overhead. 
There are other control-free algorithms proposed in Refs. \cite{LuBanulsCirac2020algorithms,OBrienEtAl2020error,LinTong2022} for energy and phase estimation via the measurement of certain scalar expectation values as the output. In particular, such algorithms cannot coherently implement a controlled time evolution and are therefore not compatible with the implementation of \QETU. In this paper, we exploit certain anti-commutation relations and structures of the Hamiltonian to propose a new control-free implementation. In the context of \QETU, the algorithm does not introduce any ancilla qubit and requires a small number of two-qubit gates that scales linearly in $n$.
We demonstrate the optimized circuit implementation of the transverse field Ising model under the control-free setting. 
To the extent of our knowledge, this circuit is significantly simpler than all previous QSP-type circuits for simulating a physical Hamiltonian.  We show the numerical performance of our algorithm for estimating the ground energy estimation in the presence of tunable quantum error using \textsf{IBM Qiskit}.

\begin{table}[th]
\label{tab:energy_estimation}
    \centering
    \makegapedcells
        \begin{tabular}{p{4cm}|c|c||c|c|c}
        \hline
        \hline
                    & Query  & Query  & \# ancilla &  Need & Input\\
                    & depth & complexity  &  qubits & MQC? & model\\
        \hline
        This work (Theorem \ref{thm:ground_state_energy_woaa} )    & $\wt{\Or}(\epsilon^{-1})$ & $\wt{\Or}(\epsilon^{-1}\gamma^{-2})$ & $\Or(1)$ & No & HE\\
        \hline
        This work (Theorem \ref{thm:ground_state_energy})    & $\wt{\Or}(\epsilon^{-1}\gamma^{-1})$ & $\wt{\Or}(\epsilon^{-1}\gamma^{-1})$ & $\Or(1)$ & Low & HE\\
        \hline
        QPE (high confidence) \cite{knill2007optimal,poulin2009sampling,nagaj2009fast} & $\wt{\Or}(\epsilon^{-1})$ & $\wt{\Or}(\epsilon^{-1}\gamma^{-2})$ & $\Or(\polylog(\gamma^{-1}\epsilon^{-1}))$ & High & HE\\
        \hline
        QPE (semi-classical) \cite{HigginsBerryEtAl2007,BerryHiggins2009} & $\wt{\Or}(\epsilon^{-1}\gamma^{-2})$ &  $\wt{\Or}(\epsilon^{-1}\gamma^{-4})$ & $\Or(1)$ & No & HE\\
        \hline
        QEEA \cite{somma2019quantum,LinTong2022} & $\wt{\Or}(\epsilon^{-1})$ & $\wt{\Or}(\epsilon^{-4}\gamma^{-4})$ & $\Or(1)$ & No & HE \\
        \hline
        GTC19 (Theorem 4) \cite{ge2019faster} & $\wt{\Or}(\epsilon^{-3/2}\gamma^{-1})$ & $\wt{\Or}(\epsilon^{-3/2}\gamma^{-1})$ & $\Or(\log(\epsilon^{-1}))$ & High & HE \\
        \hline
        LT20 \cite{LinTong2020a} & $\wt{\Or}(\epsilon^{-1}\gamma^{-1})$ & $\wt{\Or}(\epsilon^{-1}\gamma^{-1})$ & $m+\Or(\log(\epsilon^{-1}))$ & High & BE\\
        \hline
        LT22 \cite{LinTong2022} & $\wt{\Or}(\epsilon^{-1})$ & $\wt{\Or}(\epsilon^{-1}\gamma^{-4})$ & $\Or(1)$ & No & HE \\
        \hline 
        \hline
        \end{tabular}
    \caption{Comparison of the performance of quantum algorithms for ground-state energy estimation in terms of the query complexity, query depth, number of ancilla qubits, and the level of multi-qubit control (abbreviated as ``MQC'' in the table) operation is needed. $\gamma$ is the overlap between the initial guess $\ket{\phi_0}$ and the ground state, and $\epsilon$ is the allowed error. ``HE'' stands for the Hamiltonian evolution model (assuming no ancilla qubits), and ``BE'' for the block encoding model. The sharper estimate for estimating the ground-state energy using the quantum eigenvalue estimation algorithm (QEEA) is given in~\cite[Appendix C]{LinTong2022}. 
We assume that Ref.~\cite{LinTong2020a} uses $m$ ancilla qubits, and the high-level MQC operation is due to the block encoding of $H$.}
    \label{tab:compare_algs_energy}
\end{table}

\begin{table}[th]
\label{tab:ground_state_prep}
    \centering
    \makegapedcells
        \begin{tabular}{p{4cm}|c|c||c|c|c}
        \hline
        \hline
                    & Query  & Query  & \# ancilla & Need & Input\\
                    & depth & complexity &  qubits & MQC? & model\\
        \hline
        This work (Theorem \ref{thm:ground_state_prep_with_bound_wo_AA})    & $\wt{\Or}(\Delta^{-1})$ & $\wt{\Or}(\Delta^{-1}\gamma^{-2})$ & $\Or(1)$ & No & HE\\
        \hline
        This work (Theorem \ref{thm:ground_state_prep_with_bound})    & $\wt{\Or}(\Delta^{-1}\gamma^{-1})$ & $\wt{\Or}(\Delta^{-1}\gamma^{-1})$ & $\Or(1)$ & Low & HE\\
        \hline
        QPE (high confidence) \cite{knill2007optimal,poulin2009sampling,nagaj2009fast} & $\wt{\Or}(\Delta^{-1})$ & $\wt{\Or}(\Delta^{-1}\gamma^{-2})$ & $\Or(\polylog(\Delta^{-1}\gamma^{-1}\epsilon^{-1}))$ & High & HE\\
        \hline
        QPE (semi-classical) \cite{HigginsBerryEtAl2007,BerryHiggins2009} & $\wt{\Or}(\Delta^{-1}\gamma^{-2})$ & $\wt{\Or}(\Delta^{-1}\gamma^{-4})$ & $\Or(1)$ & No & HE\\
        \hline
        GTC19 (Theorem 1) \cite{ge2019faster} & $\wt{\Or}(\Delta^{-1}\gamma^{-1})$ & $\wt{\Or}(\Delta^{-1}\gamma^{-1})$ & $\Or(\log(\Delta^{-1})+\log\log(\epsilon^{-1}))$  & High & HE\\
        \hline
        LT20 \cite{LinTong2020a} & $\wt{\Or}(\Delta^{-1}\gamma^{-1})$ & $\wt{\Or}(\Delta^{-1}\gamma^{-1})$ & $m$ & High & BE\\
        \hline 
        \hline
        \end{tabular}
    \caption{Comparison of the performance of quantum algorithms for ground-state preparation in terms of the query complexity, query depth, number of ancilla qubits, and the level of multi-qubit control (abbreviated as ``MQC'' in the table) operation is needed. $\gamma$ is the overlap between the initial guess $\ket{\phi_0}$ and the ground state, $\Delta$ is a lower bound of the spectral gap, and $1-\epsilon$ is the target fidelity. ``HE'' stands for the Hamiltonian evolution model (assuming no ancilla qubits), and ``BE'' the block encoding model. Here we assume that an upper bound of the ground-state energy is known ($\mu$ in Theorem \ref{thm:ground_state_prep_with_bound_wo_AA}). The algorithm in Ref. \cite{ge2019faster} (GTC19 in the table) requires precise knowledge of the ground-state energy. We assume that Ref.~\cite{LinTong2020a} uses $m$ ancilla qubits, and the high-level MQC operation is due to the block encoding of $H$.}
    \label{tab:compare_algs_state}
\end{table}

\section{Quantum eigenvalue transformation of unitary matrices}

Given the Hamiltonian evolution input model $U=e^{-i H}$, we first demonstrate that by slightly modifying the circuit for the Hadamard test in \cref{fig:main_circuits} (b), we can approximately prepare a target state $\ket{\psi_f}=f(H)\ket{\psi}/\norm{f(H)\ket{\psi}}$ efficiently and with controlled accuracy for a large class of real functions $f$. 
Specifically, this requires alternately applying the controlled forward time evolution operator $U$, a single qubit $X$ rotation in the ancilla qubit, and the controlled backward time evolution operator $U^{\dag}$ (see \cref{fig:main_circuits} (c)).  
This circuit does not store the eigenvalues of $H$ either in a classical or a quantum register, and the information of the function $f$ of interest is entirely stored in the adjustable parameters $\varphi_0,\varphi_1,\varphi_2,\ldots,\varphi_{d/2}$. 
These parameters form a set of symmetric phase factors $(\varphi_0,\varphi_1,\varphi_2,\ldots,\varphi_2,\varphi_1,\varphi_0)\in\RR^{d+1}$ used in the circuit \cref{fig:main_circuits} (c). The symmetry of the phase factors is the key to attaining the reality of the function $f$ of interest.

\begin{thm}[\QETU]
Let $U=e^{-i H}$ with an $n$-qubit Hermitian matrix $H$. 
For any even real polynomial $F(x)$ of degree $d$ satisfying  $|F(x)|\le 1, \forall x \in [-1, 1]$, we can find a sequence of symmetric phase factors $\Phi_z := (\varphi_0, \varphi_1, \cdots, \varphi_1,\varphi_0) \in \RR^{d+1}$,
such that the circuit in \cref{fig:main_circuits} (c) denoted by $\mc{U}$ satisfies $(\bra{0}\otimes I_n)\mc{U}(\ket{0}\otimes I_n)=F\left(\cos \frac{H}{2}\right)$. \label{thm:qet_unitary}
\end{thm}

The proof of \cref{thm:qet_unitary} is given in \cref{sec:qetu}.
It is worth mentioning that the concept of ``qubitization''~\cite{LowChuang2019,GilyenSuLowEtAl2019} appears very straightforwardly in \QETU.
Let the matrix function of interest be expressed as $f(H)=(f\circ g)(\cos \frac{H}{2})$, where $g(x)=2\arccos(x)$. 
Therefore we can find a polynomial approximation $F(x)$ so that 
\begin{equation}
\sup_{x\in[\sigma_{\min},\sigma_{\max}]}\abs{(f\circ g)(x)-F(x)}\le \epsilon.
\end{equation}
Here $\sigma_{\min}=\cos \frac{\lambda_{\max}}{2}, \sigma_{\max}=\cos \frac{\lambda_{\min}}{2}$, respectively (note that $\cos (x/2)$ is a monotonically \textit{decreasing} function on $[0,\pi]$).
This ensures that the operator norm error satisfies 
\[
\norm{(\bra{0}\otimes I_n)\mc{U}(\ket{0}\otimes I_n)-f(H)}\le \epsilon.
\]

The implementation of $U=e^{-i H}$ corresponds to a Hamiltonian simulation problem of $H$ at time $t=1$. In practice, we can use the Trotter decomposition to obtain an approximate implementation of $U$ without ancilla qubits, i.e., we can partition the time interval into $r$ steps with $\tau=r^{-1}$ and use a low order Trotter method to implement an approximation to $U_{\tau}\approx e^{-i H \tau}$. Then
\begin{equation}
U=e^{-i H }\approx (U_\tau)^r.
\end{equation}
In line with other works in analyzing the performance of quantum algorithms using the HE input model~\cite{HigginsBerryEtAl2007,BerryHiggins2009,somma2019quantum,LinTong2022},
in the discussion below, unless otherwise specified, we assume $U$ is implemented exactly, and the errors are due to other sources such as polynomial approximation, the binary search process etc.
We refer readers to \cref{sec:trotter} for the complexity analysis of \QETU when $U$ is implemented using a $p$-th order Trotter formula, as well as its implication in the ground-state energy estimation.

\section{ground-state energy estimation and ground-state preparation}

In this section we discuss how to estimate the ground-state energy  and to prepare the ground state within the \QETU framework. 
The  setup of the problems is as follows: \red{we assume that the Hamiltonian $H$ can be accessed through its time-evolution operator $e^{-iH}$. The goal is (1) to estimate the ground-state energy, and (2) to prepare the ground state. For the first task we assume that we have access to a good initial guess $\ket{\phi_0}$ of the ground state, i.e., $|\braket{\phi_0|\psi_0}|\geq \gamma$ where $\ket{\psi_0}$ is the ground state. For the second task, we need the additional assumption that the ground-state energy $\lambda_0$ is separated from the rest of the spectrum by a gap $\Delta$. These assumptions are stated more formally in the definitions below.}

\begin{defn}[ground-state energy estimation]
\label{defn:ground_state_energy}
Suppose we are given a Hamiltonian $H$ on $n$ qubits whose spectrum is contained in $[\eta,\pi-\eta]$ for some $\eta>0$. The Hamiltonian can be accessed through a unitary $U=e^{-iH}$. Also suppose we have an initial guess $\ket{\phi_0}$ of the ground state $\ket{\psi_0}$ satisfying $|\braket{\phi_0|\psi_0}|\geq \gamma$. This initial guess can be prepared by $U_I$. The oracles $U$ and $U_I$ are provided as black-box oracles. The goal is to estimate the ground-state energy $\lambda_0$ to within additive error $\epsilon$.
\end{defn}

\begin{defn}[ground-state preparation]
\label{defn:ground_state_prep}
Under the same assumptions as in Definition \ref{defn:ground_state_energy}, and the additional assumption that there is a spectral gap at least $\Delta$ separating the ground-state energy  $\lambda_0$ from the rest of the spectrum,
the goal is to prepare a quantum state $\ket{\wt{\psi}_0}$ such that $|\braket{\psi_0|\wt{\psi}_0}|\geq 1-\epsilon$.
\end{defn}

\red{We will primarily focus on ground-state energy estimation. This is because once we have the ground-state energy, preparing the ground state can be done by applying an approximate projection, which can be directly performed using \QETU. We will consider two settings: the short query depth setting and the near-optimal setting. In the first setting we prioritize lowering the query depth (and hence the circuit depth), and in the second setting we prioritize lowering the query complexity (and hence the total runtime). Our results for the two settings are stated in the following theorems:}
\red{
\begin{thm}[ground-state energy estimation using \QETU]
\label{thm:ground_state_energy_woaa}
Under the assumptions stated in Definition \ref{defn:ground_state_energy}, we can estimate the ground-state energy to within additive error $\epsilon$, with probability at least $1-\vartheta$, with the following cost: 
\begin{enumerate}
    \item $\wt{\Or}(\epsilon^{-1}\gamma^{-2}\log(\vartheta^{-1}))$ queries to (controlled-) $U$ and $\Or(\gamma^{-2}\polylog(\epsilon^{-1}\vartheta^{-1}))$ queries to $U_I$.
    \item One ancilla qubit.
    \item $\wt{\Or}(\epsilon^{-1}\gamma^{-2}\log(\vartheta^{-1}))$  additional one-qubit quantum gates.
    \item $\Or(\epsilon^{-1}\log(\gamma^{-1}))$ query depth of $U$.
\end{enumerate}
\end{thm}
Note that here, using the short query depth algorithm, we do not need to use extra two-qubit gates beyond what is needed in (controlled-) $U$.
\begin{thm}[near-optimal ground-state energy estimation with \QETU and improved binary amplitude estimation]
\label{thm:ground_state_energy}
Under the assumptions stated in Definition \ref{defn:ground_state_energy}, we can estimate the ground-state energy to within additive error $\epsilon$, with probability at least $1-\vartheta$, with the following cost:
\begin{enumerate}
    \item $\wt{\Or}(\epsilon^{-1}\gamma^{-1}\log(\vartheta^{-1}))$ queries to (controlled-) $U$ and $\Or(\gamma^{-1}\polylog(\epsilon^{-1}\vartheta^{-1}))$ queries to $U_I$.
    \item Three ancilla qubits.
    \item $\wt{\Or}(n\gamma^{-1}\log(\epsilon^{-1}\vartheta^{-1})+\epsilon^{-1}\gamma^{-1}\log(\vartheta^{-1}))$  additional one- and two-qubit quantum gates.
    \item $\wt{\Or}(\epsilon^{-1}\gamma^{-1}\log(\vartheta^{-1}))$ query depth of $U$.
\end{enumerate}
\end{thm}
}
\REV{To the best of our knowledge, this is also the first algorithm that can estimate the ground-state energy with $\wt{\Or}(\gamma^{-1}\epsilon^{-1})$ query complexity using only a constant number of ancilla qubits.}

\red{
We can see from the two theorems stated above that the trade-off between the query depth and the query complexity, which is also shown in Table \ref{tab:compare_algs_energy}: the short query depth algorithm has a $\wt{\Or}(\gamma^{-2})$ dependence on $\gamma$, which is sub-optimal. This is compensated by the fact that the query depth is only logarithmic in $\gamma$, which can be significantly smaller than that required in the near-optimal algorithm. Similar trade-off exists for the ground-state preparation algorithms (Theorems \ref{thm:ground_state_prep_with_bound_wo_AA} and \ref{thm:ground_state_prep_with_bound}), as shown in Table \ref{tab:compare_algs_state}.
}

We first discuss in \cref{sec:the_short_depth_algorithms} the quantum algorithms to solve these tasks with short query depth. 
As a side note, when the initial state is indeed an eigenstate of $H$, \cref{thm:ground_state_energy_woaa} also directly gives rise to a new algorithm for performing QPE using \QETU that achieves the Heisenberg-limited precision scaling (see \cref{sec:qpe_revisit}). 
Finally, assuming access to $(n+1)$-bit Toffoli gates, \cref{sec:the_near_optimal_algorithms} describes the quantum algorithms for solving the ground-state preparation and energy estimation problems with near-optimal complexity. 

\red{In Sections \ref{sec:the_short_depth_algorithms}, \ref{sec:qpe_revisit}, and \ref{sec:the_near_optimal_algorithms}, we will mostly describe the algorithms to solve these tasks and state the results as lemmas and theorems along the way. We believe this can help readers better grasp the whole picture. A exception is the proofs of Theorems \ref{thm:ground_state_energy_woaa} and \ref{thm:ground_state_energy}, which are presented as formal proofs.}

\subsection{Algorithms with short query depths}\label{sec:the_short_depth_algorithms}

Let us first focus on the ground-state preparation problem. 
We first consider a simple setting in which we assume knowledge of a parameter $\mu$ such that
\begin{equation}
\label{eq:defn_mu}
    \lambda_0\leq \mu-\Delta/2<\mu+\Delta/2\leq \lambda_1,
\end{equation}
where $\lambda_1$ is the first excited state energy.
We need to find a polynomial approximation to the shifted sign function 
\[
\theta(x-\mu)=\begin{cases}
1, & x\le \mu,\\
0, & x > \mu,
\end{cases}
\]
and the polynomial should satisfy the requirement in \cref{thm:qet_unitary}.
To this end, given a number $0<c<1$, we would like to find a real polynomial $f(x)$ satisfying
\begin{equation}
\abs{f(x)-c}\le \epsilon, \quad \forall x\in [\eta,\mu-\Delta/2]; \quad \abs{f(x)}\le \epsilon, \quad \forall x\in [\mu+\Delta/2,\pi-\eta].
\end{equation}
As will be discussed \cref{sec:convex}, it is preferable to choose $c$ to be slighter smaller than $1$ to avoid numerical overshooting. Compared to $c=1$, this has a negligible effect in practice and does not affect the asymptotic scaling of the algorithm.
Taking the cosine transformation in \cref{thm:qet_unitary} into account, we need to find a real even polynomial satisfying 
\begin{equation}
\abs{F(x)-c}\le \epsilon, \quad x\in [\sigma_{+},\sigma_{\max}]; \quad \abs{F(x)}\le \epsilon, \quad x\in [\sigma_{\min},\sigma_{-}]; \quad \abs{F(x)}\le 1, \quad x\in[-1,1],
\label{eqn:Fground_requirement}
\end{equation}
where
\begin{equation}
\sigma_{\pm}=\cos \frac{\mu \mp \Delta/2}{2}, \quad 
\sigma_{\min}=\cos \frac{\pi-\eta}{2}, \quad \sigma_{\max}=\cos\frac{\eta}{2}.
\end{equation}
Here we have used the fact that $\cos(\cdot)$ is a monotonically decreasing function on $[0,\pi/2]$.

To find such a polynomial $F(x)$, we may use the result in \cite[Corollary 7]{LowChuang2017a}, which constructs a polynomial of degree $\Or(\Delta^{-1}\log \epsilon^{-1})$ for $c=1$ and any $\mu\in[\eta,\pi-\eta]$. 
This algorithm first replaces the discontinuous shifted sign function by a continuous approximation using error functions (need to shift both horizontally and vertically, and symmetrize to get an even polynomial), and then truncates a polynomial expansion of the resulting smooth function.
The construction is specific to the shifted sign function. Its implementation relies on modified Bessel functions of the first kind, which should be carefully treated to ensure numerical stability especially when $\Delta$ is small.
In \cref{sec:convex}, we introduce a simple convex-optimization-based method for generating a near-optimal approximation, which does not rely on any analytic computation. The convex optimization procedure can be used not only to approximate the shifted sign function, but also to find polynomial approximations in a wide range of settings.
The process of obtaining the phase factors can also be streamlined using QSPPACK~\cite{DongMengWhaleyEtAl2021}. The details of this procedure is described in \cref{sec:convex}, and an example of the optimal approximate polynomial is given in \cref{fig:poly_ground}.

We can then run the \QETU circuit to apply $f(H)=F(\cos(H/2))$ to an initial guess $\ket{\phi_0}$. If $\ket{\phi_0}$ has a non-zero component in the direction of the ground state $\ket{\psi_0}$, then $f(H)$ will preserve this component up to a factor $c\approx 1$, but will suppress the orthogonal component by a factor $\epsilon\approx 0$, thus giving us a quantum state close to the ground state. This procedure does not always succeed due to the non-unitary nature of $f(H)$, and consequently we need to repeat it multiple times until we get a success. The number of repetitions needed is $\Or(\gamma^{-2}\log(\vartheta^{-1}))$ to guarantee a success probability of at least $1-\vartheta$.
The result is summarized in \cref{thm:ground_state_prep_with_bound_wo_AA}:
\begin{thm}[ground-state preparation using \QETU]
\label{thm:ground_state_prep_with_bound_wo_AA}
Under the same assumptions as in Definition \ref{defn:ground_state_prep}, with the additional assumption that we have $\mu$ satisfying \cref{eq:defn_mu}, we can prepare the ground state up to fidelity $1-\epsilon$, with probability at least $2/3$, with the following cost: 
\begin{enumerate}
    \item $\wt{\Or}(\gamma^{-2}\Delta^{-1}\log(\epsilon^{-1}))$ queries to \red{(controlled-)} $U$ and $\Or(\gamma^{-2})$ queries to $U_I$. 
    \item One ancilla qubit.
    \item $\wt{\Or}(\gamma^{-2}\Delta^{-1}\log(\epsilon^{-1}))$ additional one-qubit quantum gates.
    \item  $\wt{\Or}(\Delta^{-1}\log(\epsilon^{-1}\gamma^{-1}))$ maximal query depth of $U$.
\end{enumerate}
\end{thm}
\REV{Again, using the short query depth algorithm, we do not need to use extra two-qubit gates beyond what is needed in (controlled-) $U$.}
We can repeat the procedure multiple times to make the success probability exponentially close to $1$.
In this algorithm, to be more precise than the $\wt{\Or}$ notation used in Theorem \ref{thm:ground_state_prep_with_bound_wo_AA}, we need $\Or(\gamma^{-2}\Delta^{-1}\log(\gamma^{-1}\epsilon^{-1}))$ queries to $U$. There is a logarithmic dependence on $\gamma^{-1}$ because we need to account for subnormalization that comes from post-selecting measurement results when analyzing the error. Success of the above procedure is flagged by the measurement outcome of the ancilla qubit.

For ground-state energy estimation, our strategy is to adapt the binary search algorithm in \cite[Theorem 8]{LinTong2020a} to the current setting. In order to estimate the ground-state energy with increasing precision, we need to repeatedly solve a decision problem:

\begin{defn}[The fuzzy bisection problem]
\label{defn:the_fuzzy_bisection_problem}
Under the same assumptions as in Definition \ref{defn:ground_state_energy}, we are asked to solve the following problem: output $0$ when $\lambda_0\leq x-h$, and output $1$ when $\lambda_0\geq x+h$.
\end{defn}

Here the fuzziness is in the fact that when $x-h<\lambda_0<x+h$ we are allowed to output either $0$ or $1$. This is in fact essential for making this problem efficiently solvable. \red{Solving the fuzzy bisection problem will enable us to find the ground-state energy through binary search. We will discuss the details in the proof of Theorem \ref{thm:ground_state_energy_woaa}.}

To solve the fuzzy bisection problem, we need a real even polynomial $F(x)$ satisfying the following: 
\begin{equation}
    \label{eq:F_poly_requirement}
    \begin{aligned}
    & c-\epsilon'\leq F(x)\leq c+\epsilon',\ & x\in[\cos((x-h)/2),1] \\
    & |F(x)|\leq \epsilon',\ & x\in[0,\cos((x+h)/2)]
    \end{aligned}
\end{equation}
For asymptotic analysis, we can use the approximate sign function from \cite[Corollary 6]{LowChuang2017a}, and the degree of $F(x)$ is $\Or(h^{-1}\log(\epsilon'^{-1}))$. With a choice of $F(x)$ that satisfies the above requirements, if $\lambda_0\geq x+h$, then $\|F(\cos(H/2))\ket{\phi_0}\|\leq\epsilon'$; if $\lambda_0\leq x-h$, then 
\[
\|F(\cos(H/2))\ket{\phi_0}\|\geq \|F(\cos(\lambda_0/2))\ket{\phi_0}\| \geq (c-\epsilon')\gamma.
\]
Therefore, after choosing $\epsilon'=\gamma c/(2(\gamma+1))$, to solve the fuzzy bisection problem, we only need to distinguish between the following two cases: $\|F(\cos(H/2))\ket{\phi_0}\|\leq\epsilon'=\gamma c/(2(\gamma+1))$ or $\|F(\cos(H/2))\ket{\phi_0}\|\geq(c-\epsilon')\gamma=(\gamma+2)\gamma c/(2(\gamma+1))$. 
\red{These two cases are well separated, because
\[
\frac{(\gamma+2)\gamma c}{2(\gamma+1)} - \frac{\gamma c}{2(\gamma+1)} = \frac{\gamma c}{2}.
\]
Hence these two quantities are separated by a gap of order $\Omega(\gamma)$, which enables us to distinguish between them using a modified version of amplitude estimation, as will be discussed later.
}
A block encoding of $F(\cos(H/2))$ can be constructed using \QETU, which we denote by $U_{\mathrm{proj}}$:
\begin{equation}
    \label{eq:defn_U_proj}
    (\bra{0}\otimes I_n) U_{\mathrm{proj}} (\ket{0}\otimes I_n) = F(\cos(H/2)).
\end{equation}
Because of the estimate of the degree of $F(x)$, $U_{\mathrm{proj}}$ here uses $\Or(h^{-1}\log(\gamma^{-1}))$ queries to $U=e^{-iH}$.
All we need to do is to distinguish between the following two cases:
\[
\|(\bra{0}\otimes I)U_{\mathrm{proj}}(I\otimes U_I)(\ket{0}\ket{0^n})\|\leq \frac{\gamma c}{2(\gamma+1)}\ \text{or}\ \|(\bra{0}\otimes I)U_{\mathrm{proj}}(I\otimes U_I)(\ket{0}\ket{0^n})\| \geq \frac{(\gamma+2)\gamma c}{2(\gamma+1)}.
\]

This problem can be generalized into the following binary amplitude estimation problem:
\begin{defn}[Binary amplitude estimation]
\label{defn:binary_amplitude_estimation}
Let $W$ be a unitary acting on two registers 
(one with one qubit and the other with $n$ qubits), with the first register indicating success or failure. Let $A=\|(\bra{0}\otimes I_n)W(\ket{0}\ket{0^n})\|$ be the success amplitude. Given $0\leq \gamma_1<\gamma_2$, provided that $A$ is either smaller than $\gamma_1$ or greater than $\gamma_2$, we want to correctly distinguish between the two cases, i.e. output 0 for the former and 1 for the latter. 
\end{defn}

In \red{the context of the fuzzy bisection problem in Definition \ref{defn:the_fuzzy_bisection_problem},} we need to choose $W=U_{\mathrm{proj}}(I\otimes U_I)$, $\gamma_1=\gamma c/(2(\gamma+1))$, $\gamma_2 = (\gamma+2)\gamma c/(2(\gamma+1))$. Note that $\gamma_2/\gamma_1 = \gamma+2\geq 2$, and therefore henceforth we only consider the case where for some constant $c'$ we have $\gamma_2/\gamma_1\geq c'$.

Now we can use Monte Carlo sampling to estimate $A=\|(\bra{0}\otimes I_n)W(\ket{0}\ket{0^n})\|$.  
We will estimate how many samples are needed to distinguish whether $A\geq \gamma_2$ or $A\leq \gamma_1$. 
We implement $W\ket{0}\ket{0^n}$ and measure the first qubit, and the output will be a random variable, taking value in $\{0,1\}$, following the Bernoulli distribution and its expectation value is $1-A^2$. We denote $p_1=1-\gamma_1^2$ and $p_2=1-\gamma_2^2$. We will generate $N_s$ samples and check whether the average is larger than $p_{1/2}=(p_1+p_2)/2$ (in which case we choose to believe that $A\leq \gamma_1$), or the average is smaller than $p_{1/2}$ (in which case we choose to believe that $A\geq \gamma_2$). By the Chernoff--Hoeffding Theorem, the error probability is upper bounded by
\begin{equation}
    \max\big\{e^{-D(p_{1/2}||p_1)N_s},e^{-D(p_{1/2}||p_2)N_s}\big\},
\end{equation}
where 
\[
D(x||y)=x\log(x/y)+(1-x)\log((1-x)/(1-y))
\]
is the Kullback--Leibler divergence between Bernoulli distributions. Direct calculation, using the fact that $\gamma_2\geq c'\gamma_1$, shows that $D(p_{1/2}||p_1),D(p_{1/2}||p_2)\geq \Omega(\gamma_1^2)$. Therefore to ensure that the error probability is below $\vartheta'$, we only need to choose $N_s$ such that $e^{-\Omega(\gamma_1^2 N_s)}\leq \vartheta'$. Thus we get the scaling of $N_s=\Or(\gamma_1^{-2}\log(\vartheta'^{-1}))$.
From the above analysis, we have the following lemma.
\begin{lem}[Monte Carlo method for solving binary estimation]
\label{lem:binary_amplitude_estimation_wo_AA}
The binary amplitude estimation problem in Definition~\ref{defn:binary_amplitude_estimation}, with the additional assumption that there exists constant $c'>0$ such that $\gamma_2/\gamma_1\geq c'$, can be solved correctly with probability at least $1-\vartheta'$ by querying $W$ $\Or(\gamma_1^{-2}\log(\vartheta'^{-1}))$ times, and this procedure does not require additional ancilla qubits besides the ancilla qubits already required in $W$. The maximal query depth of $W$ is $\Or(1)$.
\end{lem}

\red{
Lemma \ref{lem:binary_amplitude_estimation_wo_AA} enables us to solve the fuzzy bisection problem stated in Definition \ref{defn:the_fuzzy_bisection_problem}.
With the tools introduced above, we can now prove Theorem \ref{thm:ground_state_energy_woaa}. }

\begin{proof}[Proof of Theorem \ref{thm:ground_state_energy_woaa}]
\red{We solve the ground-state energy estimation problem by performing a binary search, and at each search step we need to solve a fuzzy bisection problem, which we know can be done from the above discussion. Below we will discuss how the binary search works, i.e., why repeatedly solving the fuzzy bisection problem can help us find the ground-state energy.
For simplicity of the discussion we set $\eta=\pi/4$ in Definition \ref{defn:ground_state_energy}, i.e., the spectrum of the Hamiltonian is contained in the interval $[\pi/4,3\pi/4]$. 
In each iteration of the binary search we have $l$ and $r$ such that $l\leq \lambda_0 \leq r$. In the first iteration we choose $l=\pi/4$ and $r=3\pi/4$. With $l$ and $r$ we want to solve the following fuzzy bisection problem: output $0$ when $\lambda_0\leq (2l+r)/3$, and output $1$ when $\lambda_0\geq (l+2r)/3$. In other words we let $x=(l+r)/2$ and $h=(r-l)/6$ in Definition~\ref{defn:the_fuzzy_bisection_problem}. After solving this fuzzy bisection problem, if the output is $0$, then we know $l\leq \lambda_0\leq  (l+2r)/3$, and therefore we can update $r$ to be $(l+2r)/3$. Similarly if the output is $1$ we can update $l$ to be $(2l+r)/3$. In this way we get a new pair of $l$ and $r$ such that $l\leq \lambda_0\leq r$ and $r-l$ shrinks by $2/3$.}

\red{
The values of $l$ and $r$ will converge to $\lambda_0$ from both sides, and therefore when $r-l\leq 2\epsilon$, $\lambda_0$ will be within $\epsilon$ distance from $(l+r)/2$, thus giving us the ground-state energy estimate we want. This will take $\lceil\log_{3/2}(\pi\epsilon^{-1}/2)\rceil$ iterations because $r-l$ shrinks by a factor $2/3$ in each iteration and the initial value is $\pi/2$.}

\red{
In our context, we need to perform a binary amplitude estimation at each of the $\Or(\log(\epsilon^{-1}))$ steps to solve the fuzzy bisection problem, and therefore to ensure a final success probability of at least $\vartheta$ we need to choose $\vartheta'=\Theta(\vartheta/\log(\epsilon^{-1}))$ in Lemma \ref{lem:binary_amplitude_estimation_wo_AA}. Since $\gamma_1=\gamma c/(2(\gamma+1))=\Omega(\gamma)$, as discussed immediately after we introduced Definition \ref{defn:binary_amplitude_estimation}, each time we solve the binary amplitude estimation problem we need to use $W=U_{\mathrm{proj}}(I\otimes U_I)$ for $\Or(\gamma^{-2}(\log(\vartheta^{-1})+\log\log(\epsilon^{-1})))$ times. Note that each $U_{\mathrm{proj}}$ requires using $U=e^{-iH}$ for $\Or(h^{-1}\log(\gamma^{-1}))$ times. Adding up for $h$ that decreases exponentially until it is of order $\epsilon$, we can get the estimate for the cost of estimating the ground-state energy, as stated in Theorem \ref{thm:ground_state_energy_woaa}.}
\end{proof}

\subsection{Quantum phase estimation revisited}\label{sec:qpe_revisit}

As an application of the ground-state energy estimation algorithm using \QETU, let us revisit the task of performing quantum phase estimation (QPE). Assuming access to a unitary $U=e^{-iH}$ and an eigenstate $\ket{\psi}$ such that $U\ket{\psi}=e^{-i\lambda}\ket{\psi}$, the goal of the phase estimation is to estimate $\lambda$ to precision $\epsilon$. 
This can be viewed \REV{as} the ground-state energy estimation \REV{problem with an} initial overlap $\gamma=1$. Although $\lambda$ may not be the ground-state energy of $H$, other eigenvalues of $H$ do not matter because \REV{the initial state has zero overlap with other eigenstates.}

In order to estimate $\lambda$, we can repeatedly solve the fuzzy bisection problem in \cref{defn:the_fuzzy_bisection_problem}, and this gives us an algorithm that is essentially identical to the one described in the proof of Theorem \ref{thm:ground_state_energy_woaa}.
As a corollary of \cref{thm:ground_state_energy_woaa}, the phase estimation problem can be solved with the following cost that achieves the Heisenberg limit:
\begin{cor}[Phase estimation]
\label{cor:phase_estimation}
Suppose we are given $U=e^{-iH}$, a quantum state $\ket{\psi}$ such that $U\ket{\psi}=e^{-i\lambda}\ket{\psi}$, where $\lambda\in[\eta,\pi-\eta]$ for some constant $\eta>0$. We can estimate $\lambda$ to precision $\epsilon$ with probability at least $1-\vartheta$ using $\wt{\Or}(\epsilon^{-1}\log(\vartheta^{-1}))$ applications to (controlled) $U$ and its inverse, a single copy of $\ket{\psi}$, and $\Or(\epsilon^{-1}(\log(\vartheta^{-1})+\log\log(\epsilon^{-1})))$ additional one qubit gates.
 \end{cor}
 
It may require some explanation as to why we only need a single copy of $\ket{\psi}$ rather than repeatedly apply a circuit that prepares $\ket{\psi}$. This is because $\ket{\psi}$ is an eigenstate and consequently will be preserved up to a phase factor in the circuit depicted in Figure \ref{fig:main_circuits} (c). Therefore it can be reused throughout the algorithm and there is no need to prepare it more than once.

\subsection{Algorithms with near-optimal query complexities}
\label{sec:the_near_optimal_algorithms}

With a block encoding input model, Theorems 6 and 8 in Ref.~\cite{LinTong2020a} \REV{are} near-optimal algorithms for preparing the ground state and for estimating the ground-state energy, respectively. 
In this section, we combine \QETU with amplitude amplification and a new binary amplitude estimation method to yield quantum algorithms with the same near-optimal query complexities. For amplitude estimation we avoid using quantum Fourier transform as it would require an additional register of qubits. Instead we use a procedure described in Appendix \ref{sec:binary_amplitude_estimation} based on \QETU.
Unlike previous near-term methods for amplitude estimation \cite{wang2019accelerated,WangKohEtAl2021minimizing} that typically rely on Bayesian inference techniques, and thus require knowledge of a prior distribution, our method does not require such prior knowledge. One could also adapt the QFT-free approximate counting algorithms in \cite{Wie2019simpler,AaronsonRall2020quantum} to the amplitude estimation problem, but our approach in Appendix \ref{sec:binary_amplitude_estimation} is better tailored for the \QETU framework.

We need to use amplitude amplification to quadratically improve the $\gamma$ dependence in Theorem \ref{thm:ground_state_prep_with_bound_wo_AA}, and to achieve the near-optimal query complexity for preparing the ground state in Definition \ref{defn:ground_state_prep} (assuming knowledge of $\mu$). Let us first study the number of ancilla qubits needed for this task.
For amplitude amplification we need to construct a reflection operator around the initial guess $\ket{\phi_0}$. This requires implementing $2\ket{0^n}\bra{0^n}-I$, which is equivalent, using phase kickback, to implementing an $(n+1)$-bit Toffoli gate:
\[
\ket{1^n}\bra{1^n}\otimes \sigma_x + (I-\ket{1^n}\bra{1^n})\otimes I.
\]
This $(n+1)$-bit Toffoli gate can be implemented using $\Or(n)$ elementary one- or two-qubit gates, on $n+2$ qubits \cite[Corollary 7.4]{BarencoBennettEtAl1995elementary}.
Note that this can be a relatively costly operation on early fault-tolerant quantum devices. We need two ancilla qubits to implement the reflection operator, but one of them can be reused for other purposes. The reason is as follows: one ancilla qubit is the one that $\sigma_x$ acts on conditionally in the $(n+1)$-bit Toffoli gate. This one cannot be reused because it needs to start from $\ket{0}$ and will be returned to $\ket{0}$. The other qubit, however, can start from any state and will be returned to the original state, as discussed in \cite[Corollary 7.4]{BarencoBennettEtAl1995elementary}, and therefore we can use any qubit in the circuit for this task, except for the $n+1$ qubits already involved in the $(n+1)$-bit Toffoli gate. Note in Theorem \ref{thm:ground_state_prep_with_bound_wo_AA} we have one ancilla qubit that is used for \QETU. This qubit can therefore serve as the ancilla qubit needed in implementing the $(n+1)$-bit Toffoli gate. Thus we only need two ancilla qubits in the whole procedure.

To summarize the cost, we have the following theorem: 
\begin{thm}[near-optimal ground-state preparation with \QETU and amplitude amplification]
\label{thm:ground_state_prep_with_bound}
Under the same assumptions as in Definition \ref{defn:ground_state_prep}, with the additional assumption that we have $\mu$ satisfying \eqref{eq:defn_mu}, we can prepare the ground state, with probability $2/3$, up to fidelity $1-\epsilon$ with the following cost:
\begin{enumerate}
    \item $\wt{\Or}(\gamma^{-1}\Delta^{-1}\log(\epsilon^{-1}))$ queries to \red{(controlled-)} $U$ and $\Or(\gamma^{-1})$ queries to $U_I$.
    \item Two ancilla qubits.
    \item $\wt{\Or}(n\gamma^{-1}\Delta^{-1}\log(\epsilon^{-1}))$ additional one- and two-qubit quantum gates.
    \item $\wt{\Or}(\gamma^{-1}\Delta^{-1}\log(\epsilon^{-1}))$ query depth for $U$.
\end{enumerate}
\end{thm}
Note that we can repeat this procedure multiple times to make the success probability exponentially close to $1$.

For ground-state energy estimation, in the algorithm described in the proof of Theorem \ref{thm:ground_state_energy_woaa}, our short query depth algorithm has a $\wt{\Or}(\gamma^{-2})$ scaling because of the Monte Carlo sampling in the binary amplitude estimation (Definition \ref{defn:binary_amplitude_estimation}) procedure. Here we will improve the scaling to $\wt{\Or}(\gamma^{-1})$ using the technique developed in \cite[Lemma 7]{LinTong2020a}.
The technique in \cite[Lemma 7]{LinTong2020a} uses phase estimation, and requires $\Or(\log((\gamma_2-\gamma_1)^{-1}))=\Or(\log(\gamma^{-1}))$ ancilla qubits.
In Appendix \ref{sec:binary_amplitude_estimation}, we propose a method to solve the binary amplitude estimation problem using \QETU, which reduces the number of additional ancilla qubits down from $\Or(\log(\gamma^{-1}))$ to two. The result is summarized here:
\begin{lem}[Binary amplitude estimation]
\label{lem:binary_amplitude_estimation}
The binary amplitude estimation problem in Definition~\ref{defn:binary_amplitude_estimation} can be solved correctly with probability at least $1-\vartheta'$ by querying $W$ $\Or((\gamma_2-\gamma_1)^{-1}\log(\vartheta'^{-1}))$ times, and this procedure requires one additional ancilla qubit (besides the ancilla qubits already required in $W$).
\end{lem}

The key idea is to treat the walk operator in amplitude estimation as a time evolution operator corresponding to a Hamiltonian, and this allows us to apply \QETU to extract information about that Hamiltonian.

As a result of this new method to solve the binary amplitude estimation problem, the ground-state energy estimation problem can now be solved using only three ancilla qubits: one for \QETU, and two others for binary amplitude estimation.

\red{With these results we can now analyze the cost of ground-state energy estimation in the near-optimal setting, and thereby prove Theorem \ref{thm:ground_state_energy}.}
\begin{proof}[Proof of Theorem \ref{thm:ground_state_energy}]
\red{We adopt the same strategy of performing a binary search to locate the ground-state energy, as used in Theorem \ref{thm:ground_state_energy_woaa}. The main difference is that instead of using Monte Carlo sampling to solve the binary amplitude estimation problem, we now use \QETU to do so, with the complexity stated in Lemma \ref{lem:binary_amplitude_estimation}.}

\red{We now count how many times we need to query $U=e^{-iH}$ in this approach. Each time we perform binary amplitude estimation, we need to use $U_{\mathrm{proj}}$ for $\Or(\gamma^{-1}\log(\vartheta'^{-1}))$ times to have at least $1-\vartheta'$ success probability each time. We need to perform binary amplitude estimation for each step of the binary search, and there are in total $\Or(\log(\epsilon^{-1}))$ steps. Therefore to ensure a final success probability of at least $1-\vartheta$ we need to choose $\vartheta'=\Theta(\vartheta/\log(\epsilon^{-1}))$. At the $k$-th binary search step, $U_{\mathrm{proj}}$ uses $U=e^{-iH}$ for $\Or((3/2)^k \log(\gamma^{-1}))$ times. Therefore in total we need to query $U$, up to a constant factor
\[
\sum_{k=0}^{\lceil\log_{3/2}(\pi\epsilon^{-1}/2)\rceil} (3/2)^k \log(\gamma^{-1})\gamma^{-1}\log(\vartheta'^{-1}) = \Or\Big(\epsilon^{-1}\gamma^{-1}\log(\gamma^{-1})\big(\log(\vartheta^{-1})+\log\log(\epsilon^{-1})\big)\Big)
\]
times. This query complexity agrees with that in \cite[Theorem 8]{LinTong2020a} up to a logarithmic factor.
}

\red{When we count the number of additional quantum gates needed, there will be an $n$ dependence, which comes from the fact that we need to implement a reflection operator $2\ket{0^{n+1}}\bra{0^{n+1}}-I$ each time we implement $U_{\mathrm{proj}}$ in the binary amplitude estimation procedure (see Appendix \ref{sec:binary_amplitude_estimation}). These reflection operators require $\wt{\Or}(n\gamma^{-1}\log(\epsilon^{-1}\vartheta^{-1}))$ gates. We also need $\Or(\epsilon^{-1}\gamma^{-1}\log(\vartheta^{-1}))$ additional quantum gates that come from implementing \QETU. Combining the two numbers we get the number of gates as shown in the theorem.}
\end{proof}

\red{When preparing the ground state, the parameter $\mu$ in Theorem \ref{thm:ground_state_prep_with_bound} is generally not known \textit{a priori}. For exactly solvable models or small quantum systems, despite the ability to simulate them classically, one might still want to prepare the ground state, and in such cases $\mu$ is available. In the general case, to} prepare the ground state without knowing a parameter $\mu$ as in Theorem \ref{thm:ground_state_prep_with_bound}, we can first estimate the ground-state energy to within additive error $\Or(\Delta)$, and then run the algorithm in Theorem \ref{thm:ground_state_prep_with_bound}. This results in an algorithm with the following costs:
\begin{thm}\label{thm:ground_state_prep_full}
Under the assumptions stated in Definition \ref{defn:ground_state_prep}, we can prepare the ground state to fidelity at least $1-\epsilon$, with probability at least $1-\vartheta$, with the following cost:
\begin{enumerate}
    \item $\wt{\Or}(\Delta^{-1}\gamma^{-1}\polylog(\epsilon^{-1}\vartheta^{-1}))$ queries to \red{(controlled-)} $U$ and $\Or(\gamma^{-1}\polylog(\Delta^{-1}\epsilon^{-1}\vartheta^{-1}))$ queries to $U_I$.
    \item Three ancilla qubits.
    \item $\wt{\Or}(n\gamma^{-1}\log(\Delta^{-1}\epsilon^{-1}\vartheta^{-1})+\Delta^{-1}\gamma^{-1}\log(\epsilon^{-1}\vartheta^{-1}))$ additional one- and two-qubit quantum gates. 
    \item $\wt{\Or}(\Delta^{-1}\gamma^{-1}\polylog(\epsilon^{-1}\vartheta^{-1}))$ query depth of $U$.
\end{enumerate}
\end{thm}

\section{Convex-optimization-based method for constructing approximating polynomials}
\label{sec:convex}

To approximate an even target function using an even polynomial of degree $d$, we can express the target polynomial as the linear combination of Chebyshev polynomials with some unknown coefficients $\{c_k\}$:
\begin{equation}
F(x)=\sum_{k=0}^{d/2} T_{2k}(x) c_k.
\end{equation}
To formulate this as a discrete optimization problem, we first discretize $[-1,1]$ using $M$ grid points (e.g., roots of Chebyshev polynomials $\{x_j=-\cos\frac{j\pi}{M-1}\}_{j=0}^{M-1}$). One can also restrict them to $[0,1]$ due to symmetry).  We define the coefficient matrix, $A_{jk}=T_{2k}(x_j), \quad k=0,\ldots,d/2$. Then the coefficients for approximating the shifted sign function can be found by solving the following optimization problem
\begin{equation}
\begin{split}
\min_{\{c_k\}} \quad& \max\left\{\max_{x_j\in[\sigma_{\max},\sigma_{+}]} \abs{F(x_j)-c},\max_{x_j\in[\sigma_{\min},\sigma_{-}]} \abs{F(x_j)}\right\}\\
\text{s.t.} \quad & F(x_j)=\sum_{k} A_{jk} c_k, \quad \abs{F(x_j)}\le c, \quad \forall j=0,\ldots,M-1.
\end{split}
\label{eqn:opt_ground}
\end{equation}
This is a convex optimization problem and can be solved using software packages such as CVX~\cite{cvx}. 
The norm constraint $\abs{F(x)}\le 1$ is relaxed to $\abs{F(x_j)}\le c$ to take into account that the constraint can only be imposed on the sampled points, and the values of $|F(x)|$ may slightly overshoot on $[-1,1]\backslash \{x_j\}_{j=0}^{M-1}$. The effect of this relaxation is negligible in practice and we can choose $c$ to be sufficiently close to $1$ (for instance, $c$ can be $0.999$). Since \cref{eqn:opt_ground} approximately solves a min-max problem, it achieves the near-optimal solution (in the sense of the $L^{\infty}$ norm) by definition both in the asymptotic and pre-asymptotic regimes.

Once the polynomial $F(x)$ is given, the Chebyshev coefficients can be used as the input to find the symmetric phase factors using an optimization-based method. 
\REV{Due to the parity constraint, the number of degrees of freedom in the target polynomial $F(x)$ is $\wt{d} := \lceil \frac{d+1}{2} \rceil$. Hence $F(x)$ is entirely determined by its values on $\wt{d}$ distinct points.
We may choose these points to be $x_k=\cos\left(\frac{2k-1}{4\wt{d}}\pi\right)$, $k=1,...,\wt{d}$, which are the positive nodes of the Chebyshev polynomial $T_{2 \wt{d}}(x)$.
The problem of finding the symmetric phase factors can be equivalently solved via the following optimization problem
\begin{equation}\label{eqn:optprob}
    \Phi^* = \argmin_{\substack{\Phi \in [-\pi,\pi)^{d+1},\\
    \text{symmetric.}}} \mc{F}(\Phi),\ \mc{F}(\Phi) := \frac{1}{\wt{d}} \sum_{k=1}^{\wt{d}} \abs{g(x_k, \Phi) - F(x_k)}^2,
\end{equation}
where
\[
g(x,\Phi):=\Re[\braket{0|e^{\I \phi_0 Z} e^{\I \arccos(x) X} e^{\I \phi_1 Z} e^{\I \arccos(x) X} \cdots e^{\I \phi_{d-1} Z} e^{\I \arccos(x) X} e^{\I \phi_d Z}|0}].
\]
The desired phase factor achieves the global minimum of the cost function with $\mc{F}(\Phi^*)=0$. 
It has been found that a quasi-Newton method to solve \cref{eqn:optprob} with a particular symmetric initial guess 
\begin{equation}\label{eqn:phi0}
\Phi^0=(\pi/4,0,0,\ldots, 0,0,\pi/4),
\end{equation}
can robustly find the symmetric phase factors. 
Although the optimization problem is highly nonlinear, the success of the optimization based algorithm can be explained in terms of the strongly convex energy landscape near $\Phi^0$~\cite{WangDongLin2021}. 
Numerical results indicate that on a laptop computer, CVX can find the near-optimal polynomials for $d\sim 5000$. 
Given the target polynomial, the optimization based algorithm can  find phase factors for $d\sim 10000$~\cite{DongMengWhaleyEtAl2021}.
This should be more than sufficient for most QSP-based applications on early fault-tolerant quantum computers.}
The streamlined process of finding near-optimal polynomials and the associated phase factors has been implemented in QSPPACK~\footnote{\url{https://github.com/qsppack/QSPPACK}}.

As an illustrative example of the numerically optimized min-max polynomials, we set $\eta=0.1,\mu=1.0,\Delta=0.4,M=400,c=0.999$. This corresponds to $\sigma_{\min}=0.0500,\sigma_{-}=0.8253,\sigma_{+}=0.9211,\sigma_{\max}=0.9997$.
The resulting polynomial and the pointwise errors with $d=20$ and $d=80$ are shown in \cref{fig:poly_ground}. 
We remark that the polynomial has been reconstructed using the numerically optimized phase factors using QSPPACK, and hence has taken the error in the entire process into account.
We find that the pointwise error of the polynomial approximation satisfies the equioscillation property on each of the intervals $[\sigma_{\min},\sigma_{-}]$,$[\sigma_{+},\sigma_{\max}]$. This resembles the Chebyshev equioscillation theorem of the best polynomial approximation on a single interval (see e.g. ~\cite[Chapter 10]{Trefethen2019Approximation}).
\cref{fig:poly_ground_deg_conv} shows that the maximum pointwise error on the desired intervals converges exponentially with the increase of the polynomial degree. 

\begin{figure}
\begin{center}
\subfloat[]{\includegraphics[width=0.3\textwidth]{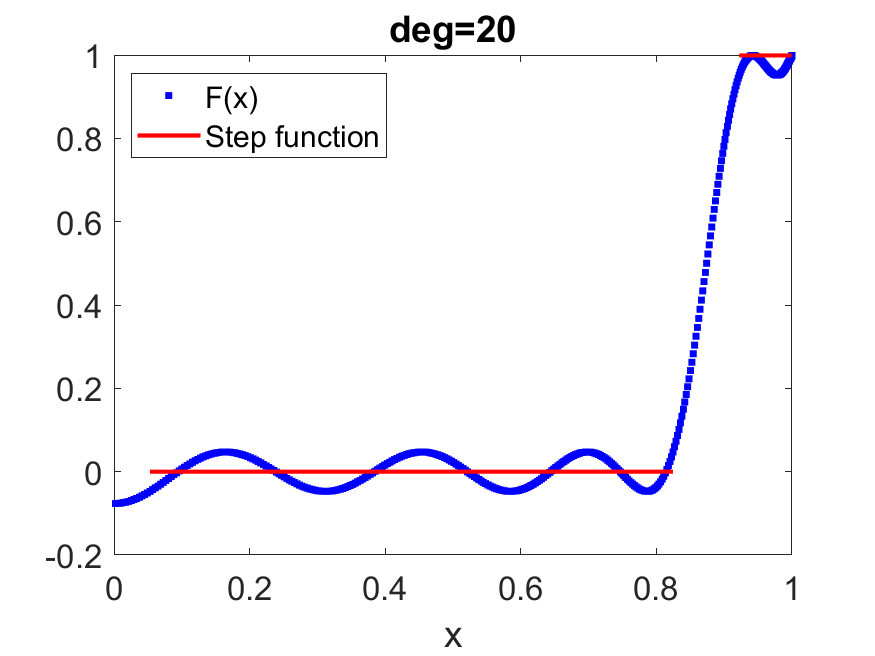}}
\subfloat[]{\includegraphics[width=0.3\textwidth]{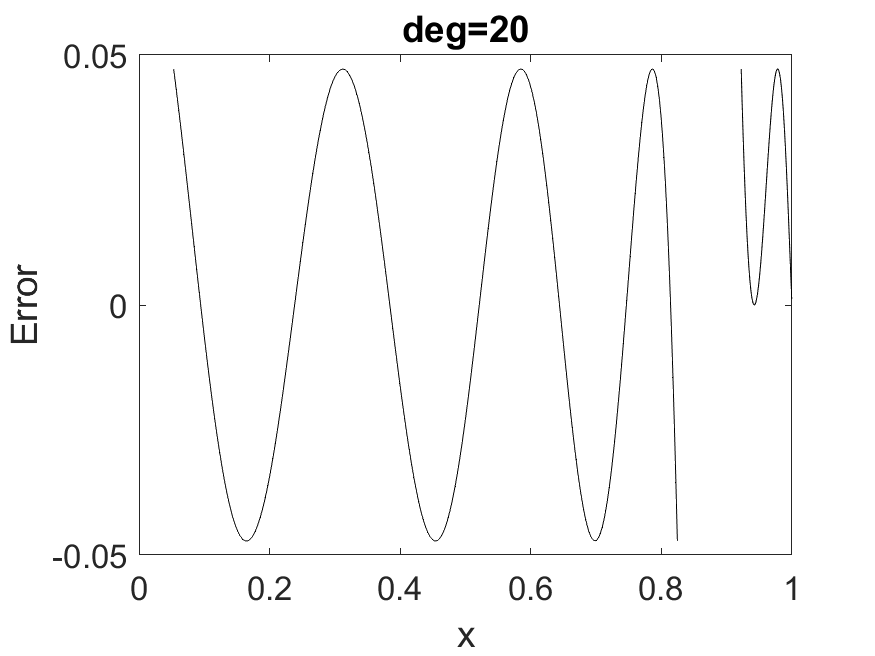}}

\subfloat[]{\includegraphics[width=0.3\textwidth]{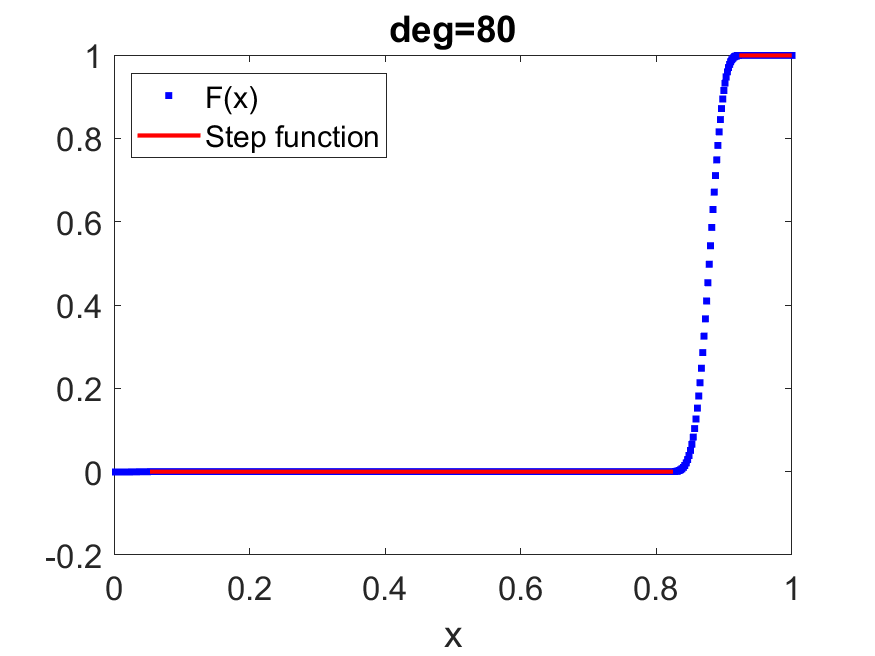}}
\subfloat[]{\includegraphics[width=0.3\textwidth]{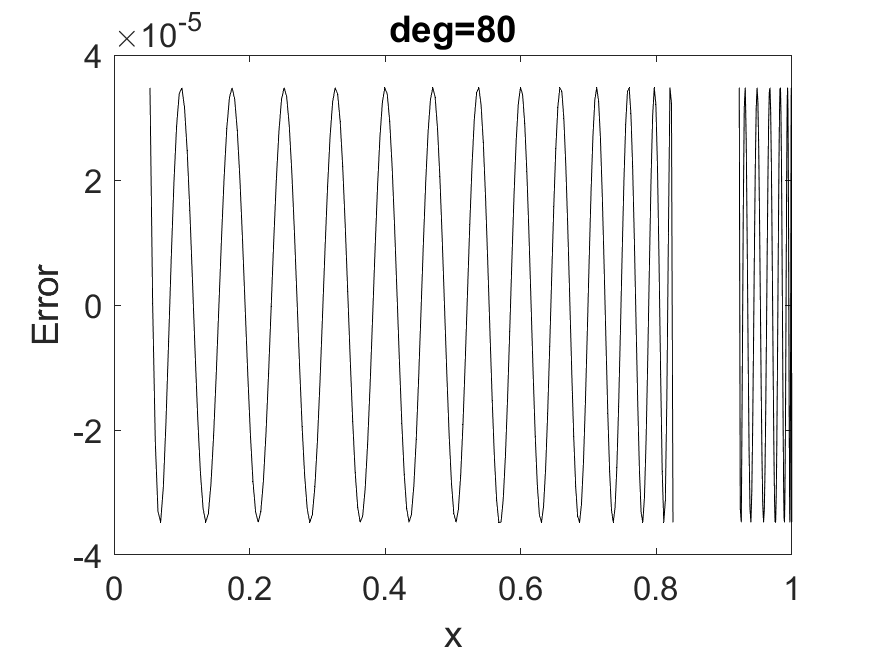}}
\end{center}
\caption{The polynomial obtained by convex optimization approximating the shifted sign function and the pointwise error on $[\sigma_{\min},\sigma_{-}]\cup[\sigma_{+},\sigma_{\max}]$ with $d=20$ (a,b) and $d=80$ (c,d).}
\label{fig:poly_ground}
\end{figure}

\begin{figure}
\begin{center}
\includegraphics[width=0.5\textwidth]{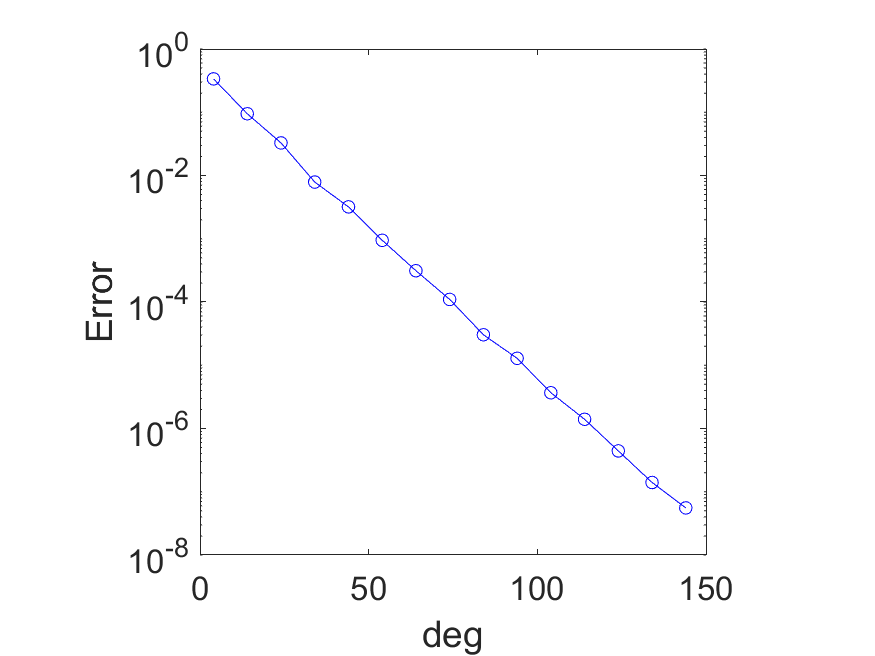}
\end{center}
\caption{Exponential convergence of the maximum pointwise error on $[\sigma_{\min},\sigma_{-}]\cup[\sigma_{+},\sigma_{\max}]$ with respect to the increase of the polynomial degree obtained by the convex optimization method.}
\label{fig:poly_ground_deg_conv}
\end{figure}

More generally, to find a min-max polynomial approximation to a general even target function $h(x)$ on a set $\mc{I}\subseteq[-1,1]$ satisfying $\abs{h(x)}\le c<1,x\in\mc{I}$ we may solve the optimization problem
\begin{equation}
\begin{split}
\min_{\{c_k\}} \quad& \max_{x_j\in \mc{I}} \abs{F(x_j)-h(x_j)}\\
\text{s.t.} \quad & F(x_j)=\sum_{k} A_{jk} c_k, \quad \abs{F(x_j)}\le c, \quad \forall j=0,\ldots,M-1.
\end{split}
\label{eqn:opt_general}
\end{equation}
We also remark that even though \QETU only concerns even polynomials, the same strategy can be applied if the target function $h$ is odd, or does not have a definite parity.

\section{Numerical comparison with QPE for ground-state energy estimation}

In this section we compare the numerical performance of our ground-state energy estimation algorithm in Theorem \ref{thm:ground_state_energy_woaa} with the quantum phase estimation algorithm implemented using semi-classical Fourier transform \cite{griffiths1996semiclassical} to save the number of ancilla qubits, as done in Refs. \cite{HigginsBerryEtAl2007,BerryHiggins2009}. We will evaluate how many queries to $U=e^{-iH}$ are needed in both algorithms to reach the target accuracy $\epsilon\le 10^{-3}$. The Hamiltonian $H$ used here has a randomly generated spectrum and is a $200\times 200$ matrix. The initial state $\ket{\phi_0}$ is guaranteed to satisfy $|\braket{\phi_0|\psi_0}|\geq \gamma$ with a tunable value of $\gamma$.  In our algorithm, the number of queries is counted by adding up the degrees of all the polynomials we need to implement using \QETU. 

\begin{figure}[htbp]
    \centering
    \subfloat[]{
    \includegraphics[width=0.514\textwidth]{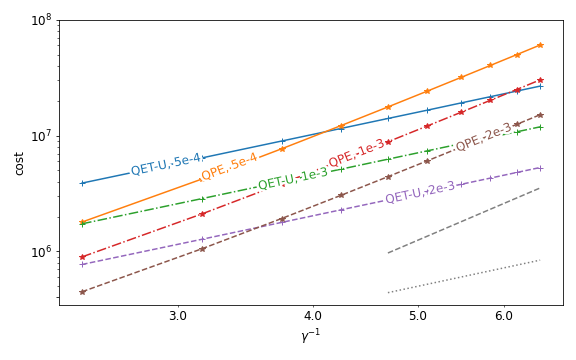}
    }
    \subfloat[]{
    \includegraphics[width=0.386\textwidth]{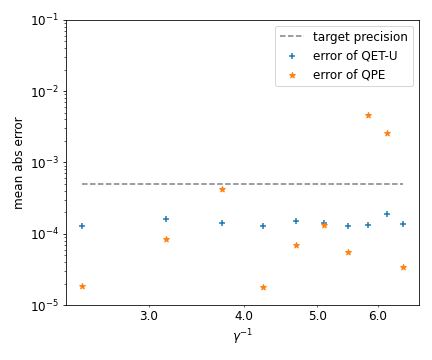}
    }
    \caption{Comparing the performance of the algorithm in Theorem \ref{thm:ground_state_energy_woaa} and the single-ancilla qubit quantum phase estimation using semi-classical Fourier transform \cite{HigginsBerryEtAl2007,BerryHiggins2009,griffiths1996semiclassical}. (a) The number of queries of $U$ needed to reach target precision \red{$\epsilon=5\times 10^{-4},10^{-3},2\times 10^{-3}$} for different values of $\gamma$. \red{The gray dashed line and dotted line show $\gamma^{-4}$ scaling and $\gamma^{-2}$ scaling respectively. Both axes are in logarithmic scale.} (b) The mean absolute error achieved by the two algorithms \COR{for target accuracy $\epsilon=5\times 10^{-4}$.}}
    \label{fig:compare_performance}
\end{figure}

\cref{fig:compare_performance} shows that to achieve comparable accuracy, our algorithm uses significantly fewer queries than quantum phase estimation, in terms of both \red{the} asymptotic scaling (improves from $\gamma^{-4}$ to $\gamma^{-2}$) as well as the \red{actual number of queries} \COR{for moderately small values of $\gamma$}. In \cref{fig:compare_performance}, the error of our method is computed by running the algorithm in Theorem \ref{thm:ground_state_energy_woaa} on a classical computer and comparing the output with the exact ground-state energy, and we show in the figure that the mean of the absolute error in multiple trials. In our method we need to determine the polynomial degree needed for each binary search step (or each time we solve the fuzzy bisection problem in Definition \ref{defn:the_fuzzy_bisection_problem}). This polynomial degree is determined by running the algorithm in \cref{sec:convex}, and selecting the smallest degree that provides an error below the target accuracy. In our numerical tests we require the approximation error to be below $10^{-3}$ so that it is much smaller than the squared overlap. The error of QPE is computed by sampling from the exact energy measurement output distribution, which is again simulated on a classical computer, and comparing the output with the exact ground-state energy. We also compute the absolute error for QPE in multiple trials and take the mean in Figure \ref{fig:compare_performance}. \red{The mean absolute errors in \cref{fig:compare_performance} (b) show that the advantage of our algorithm does not come from a loose error estimate for QPE, since \COR{our algorithm} reaches the target precision (\COR{$\epsilon=5\times 10^{-4}$} in this case) consistently and QPE does not achieve a higher precision than our algorithm.}

\section{Control-free implementation of quantum spin models}\label{sec:tfim}
In this section we demonstrate that for certain quantum spin models, the \QETU circuit can be simplified without the need of accessing the controlled Hamiltonian evolution.
\REV{
Consider a Hamiltonian $H$ that is a linear combination of $\poly(n)$ terms of Pauli operators. Note that two Pauli operators either commute or anti-commute. Hence for each term in the Hamiltonian, we can easily find another Pauli operator $K$ that anti-commutes with this term.
More generally, we assume that $H$ admits a grouping
\begin{equation}
    H = \sum_{j = 1}^\ell H^{(j)}, \quad H^{(j)} = \sum_{s=1}^{d_j} h_s^{(j)},
    \label{eqn:ham_grouping}
\end{equation}
where each $h_s^{(j)}$ is a weighted Pauli operator. For each $j$, we assume that there exists a single Pauli operator $K_j$ which anti-commutes with $H^{(j)}$, i.e.,  $K_j H^{(j)} K_j = - H^{(j)}$.  The number of groups $\ell$ is $\poly(n)$ in the worst case, but for many Hamiltonians in practice, $\ell$ may be much smaller. For example, it may be upper bounded by a constant (see examples below).} 

\REV{Conjugating the Pauli string on the time evolution operator flips the sign of the evolution time, i.e., $K_j e^{-i \tau H^{(j)}} K_j = e^{i \tau H^{(j)}}$. Since the time evolution of each Hamiltonian component is the building block of the Trotter splitting algorithm, the time flipping gives a simple implementation of the controlled time evolution without controlling the Hamiltonian components or their time evolution. To implement the controlled time evolution, it suffices to conjugate the circuit implementing the time evolution with the corresponding controlled Pauli string. Suppose $W_j(\tau) \approx U_j (\tau) := e^{-\I\tau H^{(j)}}$ is the quantum circuit approximately implementing the time evolution using Trotter splitting. Then, $K_j W_j(\tau) K_j = W_j(-\tau)$ approximates the reversed time evolution using the same splitting algorithm. This allows us to use \cref{cor:qet_unitary_2} (see \cref{sec:qetu}) which simplifies the circuit implementation of \QETU.}

To illustrate the control-free implementation, let us consider the transverse field Ising model (TFIM) for instance, whose Hamiltonian takes the form
\begin{equation}
    H_\TFIM = \underbrace{-\sum_{j=1}^{n-1} Z_j Z_{j+1}}_{H_\TFIM^{(1)}} \underbrace{- g \sum_{j=1}^n X_j}_{H_\TFIM^{(2)}}.
\end{equation}
Here $g > 0$ is the coupling constant. Note that a Pauli string 
\begin{equation}
    K := Y_1 \otimes Z_2 \otimes Y_3 \otimes Z_4 \otimes \cdots
\end{equation}
anti-commutes with both components of the Hamiltonian, namely $K H_\TFIM^{(j)} K = - H_\TFIM^{(j)}$ for $j = 1$ and $2$. Therefore, conjugating the Pauli string $K$ on the time evolution operator flips the sign of the evolution time, i.e., $K e^{-i \tau H_\TFIM^{(j)}} K = e^{i \tau H_\TFIM^{(j)}}$ for $j = 1$ and $2$. \REV{In the sense of \cref{eqn:ham_grouping}, we have $\ell=1$.} As a consequence, for TFIM, \cref{fig:main_circuits}(c) is equivalent to the circuit in \cref{fig:TFIM_circuits} in which the controlled time evolution is implicitly implemented by inserting controlled Pauli strings. 

It should be noted that the controlled Pauli string only requires implementing controlled single-qubit gates, rather than the controlled \REV{two-qubit} gates of the form $e^{-i \tau Z_j Z_{j+1}}$ (note that the Hamiltonian involves two qubit terms of the form $Z_j Z_{j+1}$).  If the quantum circuit conceptually queries the controlled time evolution $d$ times, the simplified circuit only inserts $2d$ controlled Pauli strings in the circuit. 
In this case, when implementing $W(\frac12)$ using \REV{several} Trotter steps, the controlled Pauli strings only need to be inserted before and after each $W(\frac12)$, but not between the Trotter layers. 
This \REV{simplified implementation gives the quantum circuit in \cref{fig:TFIM_circuits}(b). Note that in the general case where controlled Pauli strings are inserted in between Trotter layers, the number of controlled Pauli strings required for the implementation is $\Or(d \ell r)$ where $r$ is the number of Trotter steps to implement each time evolution operator. The simplified quantum circuit in \cref{fig:TFIM_circuits}(b) only uses $2d$ controlled Pauli strings. Therefore, this simplified implementation significantly reduces the cost when the number of Trotter step $r$ is large.}

\begin{figure}[htbp]
    \centering
\subfloat[]{
\includegraphics[width=0.9\textwidth]{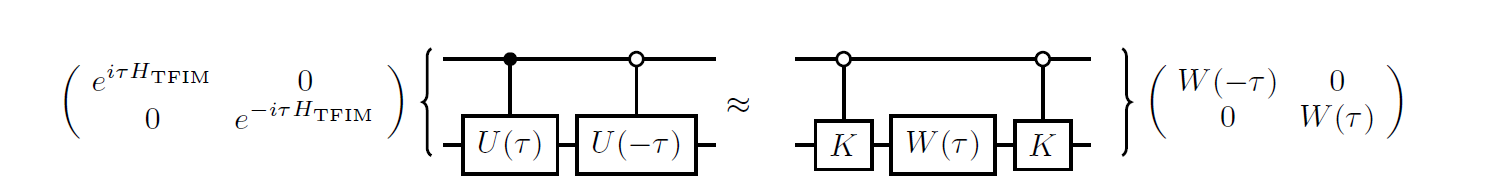}
}

\subfloat[]{
\includegraphics[width=0.9\textwidth]{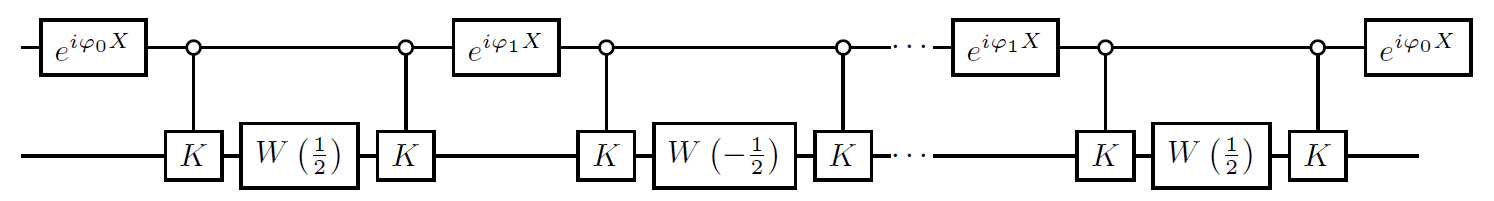}
}
    \caption{Simplified quantum circuit for simulating the \TFIM\ Hamiltonian using \QETU. (a) Controlled time evolution of the \TFIM\ Hamiltonian without directly controlling the Hamiltonian. (b) Simplified circuit for implementing \QETU. \REV{The evolution operator $W\left(-\frac{1}{2}\right)$ can also be implemented as $W\left(\frac{1}{2}\right)$ conjugated by Pauli $X$ operators (see \cref{fig:qep_unitary_circuit_2}).\label{fig:TFIM_circuits}}}
\end{figure}

\REV{In contrast to TFIM in which all Hamiltonian components share the same anti-commuting Pauli string, the control-free implementation of general spin Hamiltonian may require more Pauli strings.} For example, the Hamiltonian of the Heisenberg model takes the form
\begin{equation*}
    H_\mathrm{Heisenberg} = \underbrace{-\sum_{j=1}^{n-1} J_x X_j X_{j+1}-\sum_{j=1}^{n-1} J_y Y_j Y_{j+1}}_{H_\mathrm{Heisenberg}^{(1)}} \underbrace{- \sum_{j=1}^{n-1} J_z Z_j Z_{j+1}}_{H_\mathrm{Heisenberg}^{(2)}}.
\end{equation*}
Let us consider two Pauli strings $K_1 = Z_1\otimes I_2 \otimes Z_3 \otimes I_4 \otimes \cdots$ and $K_2 = X_1\otimes I_2 \otimes X_3 \otimes I_4 \otimes \cdots$. Then, we have the anti-commutation relations $K_1 H_\mathrm{Heisenberg}^{(1)} K_1 = - H_\mathrm{Heisenberg}^{(1)}$ and $K_2 H_\mathrm{Heisenberg}^{(2)} K_2 = - H_\mathrm{Heisenberg}^{(2)}$. Therefore, conjugating each basic time evolution $e^{-i \tau H_\mathrm{Heisenberg}^{(j)}}$ $j = 1, 2$ by controlled $K_1$ or $K_2$ respectively, we can implement the controlled time evolution without directly controlling Hamiltonians or \REV{the corresponding time evolution operators. This corresponds to $\ell=2$ in \cref{eqn:ham_grouping}.} 
Unlike the implementation for the \TFIM, the controlled Pauli strings cannot be \REV{canceled} between each Trotter layer. Therefore, simulating \REV{a} Heisenberg model requires additional controlled Pauli gates compared to that in the \TFIM\ simulation. 

\REV{
Other types of quantum Hamiltonians may also be mapped to spin Hamiltonians to perform control-free time evolution using the anti-commutation relation. Consider the 1D Fermi-Hubbard model of interacting fermions in a lattice:
\begin{equation*}
    H_\text{FH} = - \mu \sum_{j=1}^n \sum_{\sigma \in \{\uparrow, \downarrow\}} c_{j,\sigma}^\dagger c_{j,\sigma} + u \sum_{j=1}^n c_{j,\uparrow}^\dagger c_{j,\uparrow}c_{j,\downarrow}^\dagger c_{j,\downarrow} - t\sum_{j=1}^{n-1} \sum_{\sigma \in \{\uparrow, \downarrow\}} \left(c_{j,\sigma}^\dagger c_{j+1,\sigma} + c_{j+1,\sigma}^\dagger c_{j,\sigma}\right)
\end{equation*}
where $c_{j,\sigma}^\dagger$ and $c_{j,\sigma}$ ($\sigma \in \{\uparrow, \downarrow\}=\{0,1\}$) are creation and annihilation operators for different fermionic mode, $\mu$ is the chemical potential, $u$ is the on-site Coulomb repulsion energy, and $t$ is the hopping energy. The equivalent spin Hamiltonian can be derived by applying Jordan-Wigner transformation (see e.g., \cite{ReinerMarthalerBraumullerEtAl2016}), which gives (up to a global constant)
\begin{equation*}
    H_\text{FH,qubits} = \underbrace{\frac{1}{2}\left(\frac{1}{2}u-\mu\right) \sum_{j=1}^n\sum_{\sigma\in\{0,1\}} Z_{j,\sigma}}_{H_\text{FH,qubits}^{(1)}} + \underbrace{\frac{1}{4}u \sum_{j=1}^n Z_{j,0}Z_{j,1}}_{H_\text{FH,qubits}^{(2)}} \underbrace{- t \sum_{j=1}^{n-1} \sum_{\sigma \in \{0,1\}} \left(\Sigma^+_{j,\sigma}\Sigma^-_{j+1,\sigma} + \Sigma^+_{j+1,\sigma}\Sigma^-_{j,\sigma}\right)}_{H_\text{FH,qubits}^{(3)}}
\end{equation*}
where the subscript $(j,\sigma)$ stands for the $j$-th qubit on the $\sigma$-th chain, and $\Sigma^\pm_{j,\sigma} := X_{j,\sigma} \pm i Y_{j,\sigma}$. Let $K_1 = \otimes_{j, \sigma} X_{j,\sigma}$, $K_2 = \left(\otimes_{j=1}^n X_{j,0}\right)\otimes\left(\otimes_{j=1}^n I_{j,1}\right)$, and $K_3 = \left(\otimes_{j=\mathrm{even}} (Z_{j,0} \otimes Z_{j,1})\right)\otimes\left(\otimes_{j=\mathrm{odd}} (I_{j,0} \otimes I_{j,1})\right)$ be three Pauli strings. Then, we have the anti-commutation relations $K_j H_\text{FH,qubits}^{(j)}K_j = - H_\text{FH,qubits}^{(j)}$ for $j = 1, 2, 3$. Thus, controlling the time evolution of the spin-1/2 Fermi-Hubbard model can be implemented by controlling these Pauli strings. 
The construction above can also be generalized to 2D Fermi-Hubbard models. Note that direct Trotterization of the 2D Fermi-Hubbard model following the Jordan-Wigner transformation leads to non-optimal complexities, and the complexity can be improved via a fermionic swap networks~\cite{KivlichanMcCleanWiebeEtAl2018}. The control-free implementation of these more complex instances will be our future work.}

\REV{For the simplicity of implementation,} the energy estimation can be derived from the measurement frequencies of bit-strings using the standard variational quantum eigensolver (VQE) algorithm (see e.g., \cite{McCleanRomeroBabbushEtAl2016}). We state the algorithm for deriving energy estimation from measurement results in \cref{sec:num-detail} for completeness. \REV{Using the control-free implementation and VQE-type energy estimation, the implementation of the ground-state preparation using \QETU can be carried out efficiently on quantum hardware.}

\section{Numerical results for TFIM}\label{sec:numer}

\begin{figure}[htbp]
    \centering
    \includegraphics[width=\textwidth]{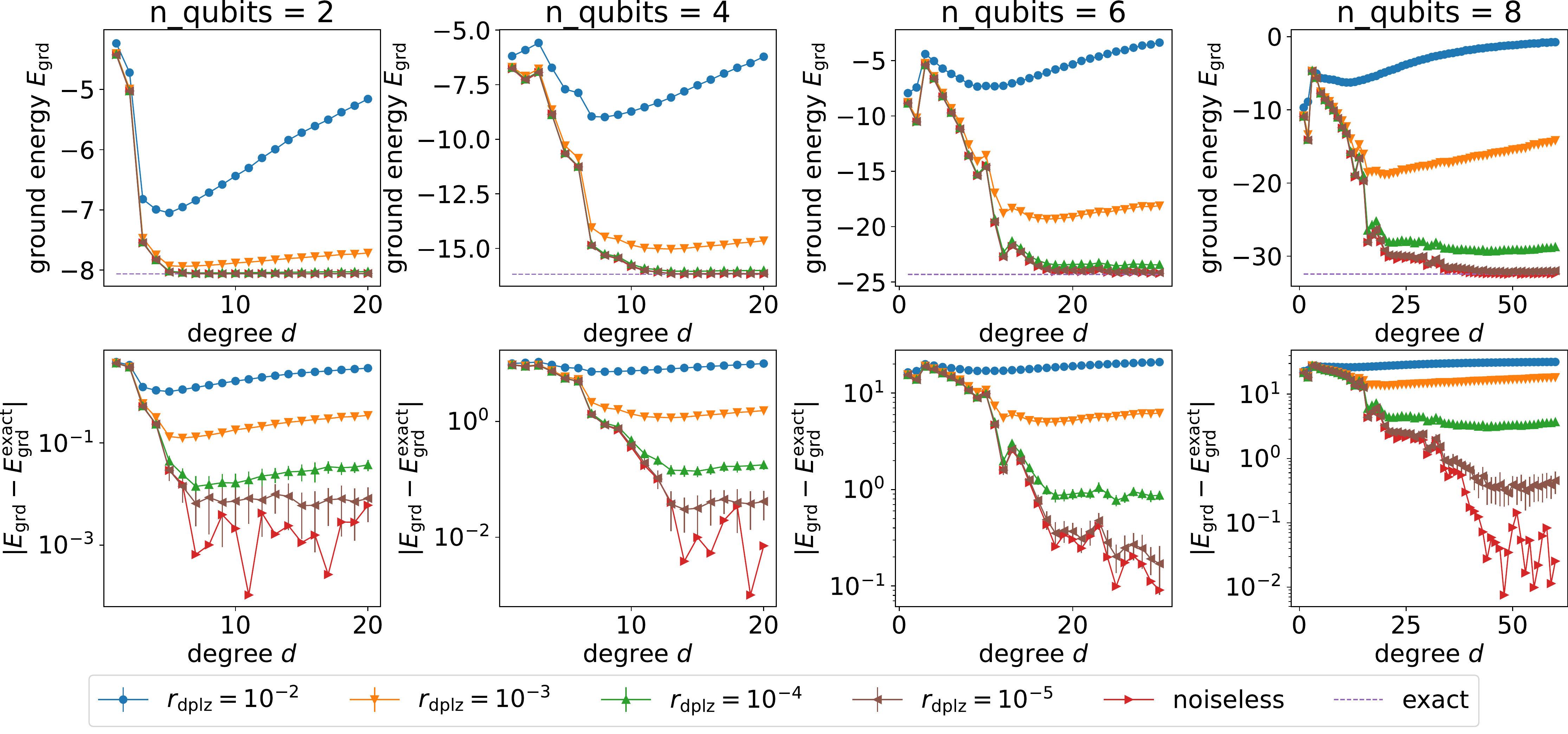}
    \caption{Estimating the ground-state energy of the TFIM model using \QETU. Each marker labels the data simulated with a given depolarizing error rate $r_\mathrm{dplz}$. The dashed line is the exact ground-state energy $E_\mathrm{grd}^\mathrm{exact}$ of the spin system. The (red) right triangles denote the data computed from the best polynomial approximation by the convex optimization solver, which only include the approximation error and are independent \REV{of} quantum noise. The error bar stands for the standard deviation estimated from $30$ repetitions.}
    \label{fig:ground_energy}
\end{figure}

\REV{Despite the potentially wide range of applications of quantum signal processing (QSP) and quantum singular value transformation (QSVT), their implementation has been limited by the large resource overhead needed to implement the block encoding of the input matrix. 
To our knowledge, QSP based quantum simulation has only been implemented for matrices encoded by random circuits~\cite{DongLin2021,DongWhaleyLin2021,CornelissenBauschGilyen2021}.
Using the \QETU and the control-free implementation in the previous section, we show that the short depth version of our algorithm for ground-state preparation and energy estimation can be readily implemented for certain physical Hamiltonians. Our implementation has a very small overhead compared to the Trotter based Hamiltonian simulation, and the circuit uses only one- and two-qubit gate operations. We demonstrate this for the TFIM via \textsf{IBM Qiskit}, and the source code is available in the Github repository~\footnote{\url{https://github.com/qsppack/QETU}}. } To demonstrate the algorithm, we prepare the ground state of the Ising model with \REV{a varying} number of qubits $n$, and coupling strength is set to $g = 4$. In the quantum circuit, we set the initial state to $\ket{0}\ket{\psi_\mathrm{in}}$ where $\ket{\psi_\mathrm{in}} = \ket{0^n}$ and the additional one qubit is the ancilla qubit on which $X$ rotations are applied. 
To simplify the numerical test, we compute the value of $\mu$ and $\Delta$ by explicitly diagonalizing the Hamiltonian. For each Hamiltonian evolution $U$ in the \QETU circuit, the number of Trotter steps is set to $r = 3$. We list \REV{system-dependent} parameters and the initial overlap $\gamma = \abs{\braket{\psi_\mathrm{in} | \psi_0}}$  for different number\REV{s} of system qubits $n$ in \cref{tab:system_params}.   According to the analysis in \cref{sec:tfim}, it is sufficient to measure two quantum circuits to estimate the ground-state energy of the TFIM. Each quantum circuit in the numerical experiment is measured with $10^5$ measurement shots, and we independently repeat the numerical test $30$ times to estimate the statistical fluctuation. 
To emulate the noisy quantum operation, we add a depolarizing error channel to each gate operation in the quantum circuit where the error of single-qubit gate operation and that of the two-qubit gate operation are set to $\frac{r_\mathrm{dplz}}{10}$ and $r_\mathrm{dplz}$ respectively. Assuming the digital error model (DEM), the total effect of the noise can be written as 
\[
    \varrho_{\mathrm{exp}} = \alpha_\mathrm{DEM} \varrho_\mathrm{exact}  + (1-\alpha_\mathrm{DEM}) \mc{E}(\varrho_\mathrm{input}),
\]
and the noise channel $\mc{E}(\varrho_\mathrm{input})$ can be modeled as a global depolarized error channel~\cite{BoixoIsakovSmelyanskiyEtAl2018} with circuit fidelity $\alpha_\mathrm{DEM}$. \REV{Given $n$ and $d$, the number of single- and two-qubit gates involved in the quantum circuit are
\begin{equation}
    \label{eqn:gate_counts}
    n_{g,1} = d(nr+1)+1 \quad \mathrm{and} \quad n_{g,2} = d\left((n-1)r+2n\right).
\end{equation}}
Therefore, the circuit fidelity can be modeled as \[
\alpha_\mathrm{DEM} = \left(1-\frac{r_\mathrm{dplz}}{10}\right)^{n_{g,1}} \left(1-r_\mathrm{dplz}\right)^{n_{g,2}}.
\] 

The numerical result is presented in \cref{fig:ground_energy}. The convergence of the noiseless data to the exact ground-state energy suggests that the energy can be computed accurately when $d$ is \REV{modest} ($10\sim 30$).  The statistical fluctuation, quantified by the standard deviation derived from $30$ repetitions, is on the order of $10^{-2}$ and is not visible in the top panels.
When simulating in the presence of the depolarizing noise, numerical results suggest that accurate estimation of the energy requires $r_\mathrm{dplz}$ to be $10^{-4}$ or less. \REV{This requirement is beyond the noise level that can achieved by current NISQ devices. Therefore we expect that QET-U based algorithms are more suited for early fault-tolerant quantum devices. For TFIM, the spectral gap $\Delta$  decreases as the number of qubits  increases. Therefore the degree of the polynomial also needs to be increased to approximate the shifted sign function and to prepare the ground state to a fixed precision (see \cref{fig:ground_energy}).}

\section{Conclusion}\label{sec:conclusion}

In this work, we develop algorithms for preparing the ground state and for estimating the ground-state energy of a quantum Hamiltonian suitable on early fault-tolerant quantum computers.
The early fault-tolerant setting limits the number of qubits, the circuit depth, and the type of multi-qubit control operations that can be employed. 
While block encoding is an elegant technique for abstractly encoding the information of an input Hamiltonian, existing block encoding strategies (such as those for $s$-sparse matrices~\cite{ChildsKothariSomma2017,GilyenSuLowEtAl2019}) can lead to a large resource overhead and cannot meet the stringent requirements of early fault-tolerant devices. The resource overhead for approximately implementing a Hamiltonian evolution input model is much lower, and can be a suitable starting point for constructing more complex quantum algorithms.

Many computational tasks can be expressed in the form of applying a matrix function $f(H)$ to a quantum state $\ket{\psi}$. 
We develop a tool called quantum eigenvalue transformation of unitary matrices with real polynomials (\QETU), which performs this task using the controlled Hamiltonian evolution as the input model (similar to that in quantum phase estimation), only one ancilla qubit and no multi-qubit control operations.
Combined with a fuzzy bisection procedure, the total query complexity of the resulting algorithm to estimate the ground-state energy scales as $\wt{\Or}(\epsilon^{-1}\gamma^{-2})$, which saturates the Heisenberg limit with target precision $\epsilon$.
The scaling with the initial overlap $\gamma$ is not optimal, but this result already outperforms all previous quantum algorithms for estimating the ground-state energy using a comparable circuit structure (see \cref{tab:compare_algs_energy}).

The \QETU technique, and the new convex-optimization-based technique for streamlining the process of finding phase factors, could readily be useful in many other contexts, such as preparing the Gibbs state. 
It is worth mentioning that other than using shifted sign functions, one can also use the exponential function $e^{-\beta (H-\mu I)}$ (the same as that needed for preparing Gibbs states, with an appropriate choice of $\beta,\mu$) to approximately prepare the ground state. This gives rise to the imaginary time evolution method. Unlike the quantum imaginary time evolution (QITE) method~\cite{MottaSunTanEtAl2020} which performs both real\REV{-}time evolution and a certain quantum state tomography procedure, \QETU  only queries the time evolution with performance guarantees and therefore can be significantly more advantageous in the early-fault-tolerant regime.

If we are further allowed to use the $(n+1)$-bit Toffoli gates (which is a relatively low-level multi-qubit operation, as the additional two\REV{-}qubit operations scale linearly in $n$), we can develop a new binary amplitude estimation algorithm that is also based on \QETU. 
The total query complexity for estimating the ground-state energy can be improved to the near-optimal scaling of $\wt{\Or}(\epsilon^{-1}\gamma^{-1})$,
at the expense of increasing the circuit depth from $\wt{\Or}(\epsilon^{-1})$ to $\wt{\Or}(\epsilon^{-1}\gamma^{-1})$. This matches the results in Ref.~\cite{LinTong2020a} with a block encoding input model.  This also provides an answer to a question raised in Ref.~\cite{LinTong2022}, i.e.,  whether it is possible to have a quantum algorithm that does not use techniques such as LCU or block encoding, with a short query depth that scales as $\wt{\Or}(\epsilon^{-1})$, and with a total query complexity that scales better than $\Or(\gamma^{-4})$. Our short query depth algorithm shows that it is possible to improve the total query complexity to $\Or(\gamma^{-2})$ while satisfying all other constraints. The construction of our near-optimal algorithm (using binary amplitude estimation) indicates that it is unlikely that one can improve the total query complexity to $\Or(\gamma^{-1})$ without introducing a factor that scales with $\gamma^{-1}$ in the circuit depth.

The improvements \REV{in} circuit depth and query complexity for preparing the ground state are similar to that of the ground-state energy estimation (see \cref{tab:compare_algs_state}). 
It is worth mentioning that many previous works using a single ancilla qubit cannot be easily modified to prepare the ground state. 
It is currently an open question whether the query complexity can be reduced to the near-optimal scaling without using any multi-qubit controlled operation (specifically, whether the additional one- and two-qubit quantum gates can be independent of the system size $n$).

In practice, the cost \REV{of} implementing the controlled Hamiltonian evolution can still be high.
By exploiting certain anti-commutation relations, we develop a new control\REV{-}free implementation of \QETU for a class of quantum spin Hamiltonians.
The results on quantum simulators using \textsf{IBM Qiskit} indicate that relatively accurate estimates to the ground-state energy can be obtained already with a \REV{modest} polynomial degree ($10\sim 30$). 
However, the results of the \QETU can be sensitive to quantum noises (such as gate-wise depolarizing noises). 
On one hand, while the \QETU circuit (especially, the control-free variant) may be simple enough to fit on a NISQ device, the error on the NISQ devices may be too large to obtain meaningful results.
On the other hand, it may be possible to combine \QETU with randomized compilation~\cite{WallmanEmerson2016} and/or error mitigation techniques~\cite{CaiXuBenjamin2020} to significantly reduce the impact of the noise, which may then enable us to obtain qualitatively meaningful results on near\REV{-}term devices~\cite{TazhigulovSunHaghshenasEtAl2022}. These will be our future works.

\vspace{1em}
\textbf{Acknowledgment}

This work was supported by the U.S. Department of Energy under the Quantum Systems Accelerator program under Grant No. DE-AC02-05CH11231 (Y.D.), by the NSF Quantum Leap Challenge Institute (QLCI) program under Grant number OMA-2016245 (Y.T.), and by the Google Quantum Research Award (L.L.). L.L. is a Simons Investigator. The authors thank \COR{Zhiyan Ding and} Subhayan Roy Moulik for discussions, \REV{and the anonymous referees for helpful suggestions}.

\widetext
\clearpage
\appendix

\begin{center}
    {\Large \bf {Supplementary Information}}
\end{center}

\section{Brief summary of polynomial matrix transformations}
\label{sec:related_matrixtrans}

In the past few years, there have been significant algorithmic advancements in efficient representation of certain polynomial matrix transformations on quantum computers~\cite{LowChuang2017,LowChuang2019,GilyenSuLowEtAl2019}, which finds applications in Hamiltonian simulation, solving linear systems of equations, eigenvalue problems, to name a few.
The commonality of these approaches is to (1) encode a certain polynomial using a product of parameterized SU(2) matrices, and (2) lift the SU(2) representation to matrices of arbitrary dimensions (a procedure called ``qubitization''~\cite{LowChuang2019} which is  related to quantum walks~\cite{Szegedy2004,Childs2010}).
This framework often leads to a very concise quantum circuit, and can unify a large class of quantum algorithms that have been developed in the literature~\cite{GilyenSuLowEtAl2019,MartynRossiTanEtAl2021}.
For clarity of the presentation, the term quantum signal processing (QSP) will specifically refer to the SU(2) representation. 
It is worth noting that the depending on the structure of the matrix and the input model, the resulting quantum circuits can be different. Block encoding~\cite{LowChuang2019,GilyenSuLowEtAl2019} is a commonly used input model for representing non-unitary matrices on a quantum computer.

\vspace{1em}

\begin{defn}[Block encoding] Given an $n$-qubit matrix $A$ ($N=2^n$), if we can find $\alpha, \epsilon \in \mathbb{R}_+$, and an $(m+n)$-qubit unitary matrix $U_A$ so that 
\begin{equation}
\Vert A - \alpha \left(\langle 0^m | \otimes I_N\right) U_A \left( | 0^m \rangle \otimes I_N \right) \Vert_2 \leq \epsilon,
\end{equation}
then $U_A$ is called an $(\alpha, m, \epsilon)$-block-encoding of $A$.
\label{def:blockencode}
\end{defn}

\vspace{1em}

When a polynomial of interest is represented by QSP, we can use the block encoding input model to implement the polynomial transformation of a Hermitian matrix, which gives the quantum eigenvalue transformation (QET)~\cite{LowChuang2019}.
Similarly the polynomial transformation of a general matrix (called singular value transformation) gives the quantum singular value transformation (QSVT)~\cite{GilyenSuLowEtAl2019}. In fact, for a Hermitian matrix with a block encoding input model, the quantum circuits of QET and QSVT can be the same.

It is worth noting that the original presentation of QSP~\cite{LowChuang2017} combines together the SU(2) representation and a trigonometric polynomial transformation of a Hermitian matrix $H$, and the input model is provided by a quantum walk operator~\cite{Childs2010}. If $H$ is a $s$-sparse matrix, the use of a walk operator is actually not necessary, and QET/QSVT gives a more concise algorithm than that in Ref.~\cite{LowChuang2017}. 

Using the Hamiltonian evolution input model, our \QETU algorithm provides a circuit structure that is similar to that in~\cite[Figure 1]{LowChuang2017}, and the derivation of \QETU is both simpler and more constructive. Note that Ref.~\cite{LowChuang2017} only states the existence of the parameterization without providing an algorithm to evaluate the phase factors, and the connection with the more explicit parameterization such as those in~\cite{GilyenSuLowEtAl2019,Haah2019} has not been shown in the literature.
Our \QETU algorithm in~\cref{thm:qet_unitary} directly connects to the parameterization in~\cite{GilyenSuLowEtAl2019}, and in particular, QSP with symmetric phase factors~\cite{DongMengWhaleyEtAl2021,WangDongLin2021}. This gives rise to a concise way for representing real polynomial transformations that is encountered in most applications.

The \QETU technique is also related to QSVT. From the Hamiltonian evolution input model $U=e^{-iH}$, we can first use one ancilla qubit and linear combination of unitaries to implement a block encoding of $\cos(H)=(U+U^{\dag})/2$. Using another ancilla qubit, we can use another ancilla qubit and QET/QSVT to implement $H=\arccos(\cos(H))$ approximately.
In other words, from the Hamiltonian evolution $U$ we can implement the matrix logarithm of $U$ to approximately block encode $H$. 
Then we can implement a matrix function $f(H)$ using the block encoding above and another layer of QSVT. \QETU simplifies the procedure above by directly querying $U$. 
The concept of ``qubitization''~\cite{LowChuang2019} appears very straightforwardly in \QETU (see \cref{sec:qetu}). 
It also saves one ancilla qubit and gives perhaps a slightly smaller circuit depth.

\section{Quantum eigenvalue transformation for unitary matrices}
\label{sec:qetu}

Let 
\begin{equation}
W(x)=e^{i \arccos(x) X}=\begin{pmatrix}
x & \I\sqrt{1-x^2}\\
\I\sqrt{1-x^2} & x
\end{pmatrix}, \quad x\in[-1,1].
\end{equation}
We first state the result of quantum signal processing for real polynomials~\cite[Corollary 10]{GilyenSuLowEtAl2018}, and specifically the symmetric quantum signal processing~\cite[Theorem 10]{WangDongLin2021} in \cref{thm:qsp_real_Wx}.

\begin{thm}[Symmetric quantum signal processing, $W$-convention]\label{thm:qsp_real_Wx}
Given a real polynomial $F(x)\in\RR[x]$, and $\deg F=d$, satisfying
\begin{enumerate}

\item $F$ has parity $d \bmod 2$,

\item $|F(x)|\le 1, \forall x \in [-1, 1]$,
\end{enumerate}
then there exists polynomials $G(x),Q(x)\in\RR[x]$ and a set of symmetric phase factors $\Phi := (\phi_0, \phi_1, \cdots, \phi_1,\phi_0) \in \RR^{d+1}$ such that the following QSP representation holds:
\begin{equation}
\begin{split}
         e^{i \phi_0 Z} \prod_{j=1}^{d} \left[ W(x) e^{i \phi_j Z} \right] = \left( \begin{array}{cc}
        F(x)+i G(x) & i Q(x) \sqrt{1 - x^2}\\
        i Q(x) \sqrt{1 - x^2} & F(x)-i G(x)
        \end{array} \right)
\end{split},
  \label{eqn:QSP_representation}
\end{equation}
\end{thm}

In order to derive \QETU, we define
\begin{equation}
W_z(x)=e^{i \arccos(x) Z}= \begin{pmatrix}
e^{i \arccos(x)} & 0 \\
0 & e^{-i \arccos(x)}
\end{pmatrix},  \quad x\in[-1,1].
\end{equation}
Then \cref{thm:qsp_real_Wz} is equivalent to \cref{thm:qsp_real_Wx}, but uses the variable $x$ is encoded in the $W_z$ matrix instead of the $W$ matrix.

\begin{thm}[Symmetric quantum signal processing, $W_z$-convention]\label{thm:qsp_real_Wz}
Given a real, even polynomial $F(x)\in\RR[x]$, and $\deg F=d$, satisfying $|F(x)|\le 1, \forall x \in [-1, 1]$, then there exists polynomials $G(x),Q(x)\in\RR[x]$ and a symmetric phase factors $\Phi_z := (\varphi_0, \varphi_1, \cdots, \varphi_1,\varphi_0) \in \RR^{d+1}$ such that the following QSP representation holds:  
\begin{equation}
\begin{split}
        U_{\Phi_z}(x) &= e^{i \varphi_0 X} W_z^*(x) e^{i \varphi_1 X} W_z(x) e^{i \varphi_2 X} \cdots e^{i \varphi_2 X} W_z^*(x) e^{i \varphi_1 X}  W_z(x) e^{i \varphi_0 X}\\
       &= \left( \begin{array}{cc}
        F(x) & -Q(x) \sqrt{1 - x^2}+i G(x)\\
        Q(x) \sqrt{1 - x^2}+i G(x)  & F(x)
        \end{array} \right).
\end{split}
  \label{eqn:QSP_Wz_real}
\end{equation}
\end{thm}
\begin{proof}
Using \begin{equation}
e^{\I\varphi X}=\mathrm{H}e^{i \varphi Z}\mathrm{H}, 
\end{equation}
and
\begin{equation}
\mathrm{H}W_z(x)\mathrm{H}=W(x), \quad \mathrm{H}W^*_z(x)\mathrm{H}=-e^{-i \frac{\pi}{2}Z}W(x)e^{-i \frac{\pi}{2}Z},
\end{equation}
we have
\begin{equation}
\begin{split}
U_{\Phi_z}(x)
=&(-1)^{\frac{d}{2}}\mathrm{H} \Big\{e^{i (\varphi_0-\frac{\pi}{2}) Z} W(x)e^{i (\varphi_1-\frac{\pi}{2}) Z} W(x) e^{i (\varphi_2-\frac{\pi}{2}) Z}\cdots \\
&e^{i (\varphi_2-\frac{\pi}{2})Z} W(x) e^{i (\varphi_1-\frac{\pi}{2}) Z} W(x) e^{i \varphi_0 Z}\Big\}\mathrm{H}\\
=&(-1)^{\frac{d}{2}} \mathrm{H}e^{-i \frac{\pi}{4} Z}
\Big\{e^{i (\varphi_0-\frac{\pi}{4}) Z} W(x)e^{i (\varphi_1-\frac{\pi}{2}) Z} W(x) e^{i (\varphi_2-\frac{\pi}{2}) Z}\cdots \\
&e^{i (\varphi_2-\frac{\pi}{2})Z} W(x) e^{i (\varphi_1-\frac{\pi}{2}) Z} W(x) e^{i (\varphi_0-\frac{\pi}{4}) Z}\Big\}e^{i \frac{\pi}{4} Z}\mathrm{H}.
\end{split}
\end{equation}
The term in the parenthesis satisfies the condition of \cref{thm:qsp_real_Wx}. We may choose a symmetric phase factor $(\phi_0,\phi_1,\ldots,\phi_1,\phi_0)$, so that
\begin{equation}
\begin{split}
&e^{i \phi_0 Z} W(x)e^{i \phi_1 Z} W(x) e^{i \phi_2 Z}\cdots e^{i \phi_2 Z} W(x) e^{i \phi_1 Z} W(x) e^{i \phi_0 Z}\\
=& (-1)^{\frac{d}{2}} \left( \begin{array}{cc}
        F(x)+i G(x) & i Q(x) \sqrt{1 - x^2}\\
        i Q(x) \sqrt{1 - x^2} & F(x)-i G(x)
        \end{array} \right).
\end{split}
\end{equation} 
Then define $\varphi_j=\phi_j+(2-\delta_{j0})\pi/4$ for $j=0,\ldots,d/2$, direct computation shows
\begin{equation}
U_{\Phi_z}(x) =\left( \begin{array}{cc}
        F(x) & -Q(x) \sqrt{1 - x^2}+i G(x)\\
        Q(x) \sqrt{1 - x^2}+i G(x) & F(x)
        \end{array} \right),
\end{equation}
which proves the theorem. 
\end{proof}

\begin{proof}[Proof of \cref{thm:qet_unitary}]
For any eigenstate $\ket{v_j}$ of $H$ with eigenvalue $\lambda_j$,
note that $\opr{span}\{\ket{0}\ket{v_j},\ket{1}\ket{v_j}\}$ is an invariant subspace of $U,U^{\dag},X\otimes I_n$, and hence of $\mc{U}$. 
Together with the fact that for any phase factors $\varphi,\varphi'$,
\begin{equation}
e^{i \varphi X} W_z^*(x) e^{i \varphi' X}  W_z(x)=
e^{i \varphi X} \begin{pmatrix}
1 & 0 \\
0 & e^{2i \arccos(x)}
\end{pmatrix} 
e^{i \varphi' X}  
\begin{pmatrix}
1 & 0 \\
0 & e^{-2i \arccos(x)}
\end{pmatrix},
\end{equation}
we have 
\begin{equation}
\mc{U}\ket{0}\ket{v_j}=(U_{\Phi_z}(\cos(\lambda_j/2))\ket{0}) \ket{v_j}=F(\cos (\lambda_j/2)) \ket{0}\ket{v_j} + \alpha_j\ket{1}\ket{v_j}.
\end{equation}
Here we have used $x=\cos(\lambda_j/2)$, and
\begin{equation}
\alpha_j=Q(\cos(\lambda_j/2))\sin(\lambda_j/2)+i G(\cos(\lambda_j/2))
\end{equation}
is an irrelevant constant according to \cref{eqn:QSP_Wz_real}.

Since any state $\ket{\psi}$ can be expanded as the linear combination of eigenstates $\ket{v_j}$ as 
\begin{equation}
\ket{\psi}=\sum_{j} c_j\ket{v_j},
\end{equation}
we have
\begin{equation}
\begin{split}
\mc{U}\ket{0}\ket{\psi}=&\sum_{j}c_j \mc{U}\ket{0}\ket{v_j}=\ket{0}\sum_{j}c_j F(\cos (\lambda_j/2)) \ket{v_j} +\ket{1}\ket{\perp}\\
=&\ket{0} F(\cos (H/2)) \ket{\psi} + \ket{1}\ket{\perp},
\end{split}
\end{equation}
where $\ket{\perp}$ is some unnormalized quantum state. 
This proves the theorem.
\end{proof}

Sometimes instead of controlled $U$, we have direct access to an oracle that simultaneously implements a controlled forward and backward time evolution:
\begin{equation}
V=\begin{pmatrix}
e^{i H} & 0 \\
0 & e^{-i H}
\end{pmatrix}.
\label{eqn:control_U}
\end{equation}
This is the case, for instance, in certain implementation of \QETU in a control-free setting. \cref{cor:qet_unitary_2} describes this version of \QETU.

\begin{figure}[H]
\begin{center}
  \includegraphics[width=0.7\linewidth]{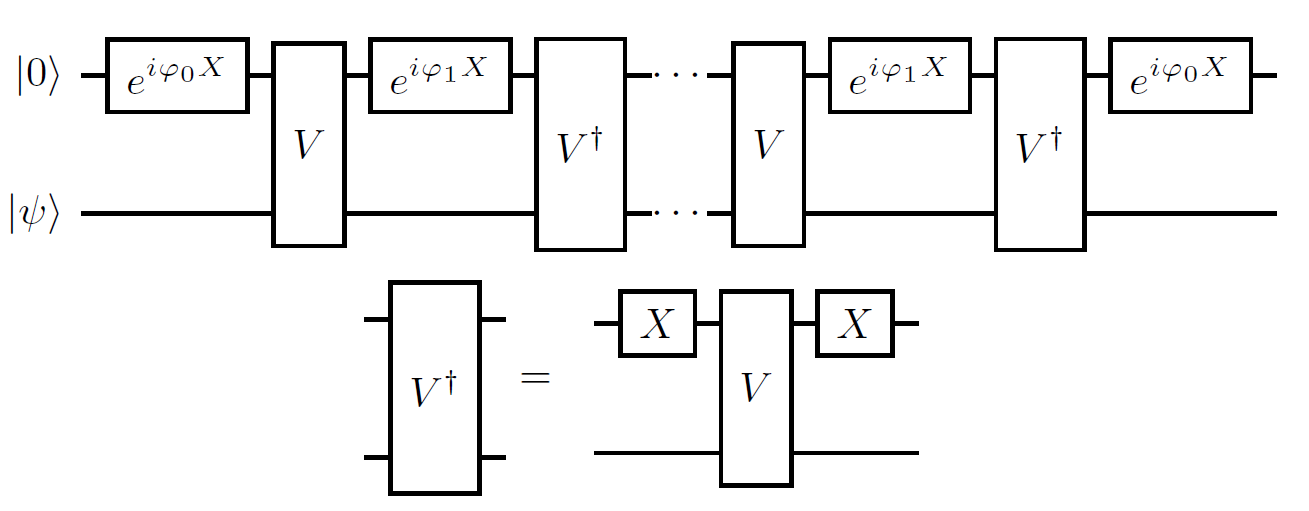}
\end{center}
  \caption{Variant of \QETU with an oracle implementing the controlled forward and backward time evolution. \REV{The implementation of $V^{\dag}$ can be carried out by conjugating $V$ with Pauli $X$ gates acting on the first qubit, and the Pauli $X$ gate can be first combined with the phase rotation as $e^{i \varphi X}X=ie^{i (\varphi+\pi/2) X}$}.}
  \label{fig:qep_unitary_circuit_2}
\end{figure}

\begin{cor}[\QETU with forward and backward time evolution]
Let $V$ be the unitary matrix given in \cref{eqn:control_U} corresponding to an $n$-qubit Hermitian matrix $H$. For any even real polynomial $F(x)$ of degree $d$ satisfying  $|F(x)|\le 1, \forall x \in [-1, 1]$, we can find a sequence of symmetric phase factors $\Phi_z := (\varphi_0, \varphi_1, \cdots, \varphi_1,\varphi_0) \in \RR^{d+1}$,
such that the circuit in \cref{fig:qep_unitary_circuit_2} denoted by $\mc{U}$ satisfies $\bra{0}\otimes I_n)\mc{U}(\ket{0}\otimes I_n)=F\left(\cos H\right)$.
\label{cor:qet_unitary_2}
\end{cor}
\begin{proof}
Let $\ket{v_j}$ be an eigenstate of $H$ with eigenvalue $\lambda_j$.
For any phase factors $\varphi,\varphi'$, using
\begin{equation}
e^{i \varphi X} W_z^*(x) e^{i \varphi' X}  W_z(x)=
e^{i \varphi X} \begin{pmatrix}
e^{-i \arccos(x)} & 0 \\
0 & e^{i \arccos(x)}
\end{pmatrix} 
e^{i \varphi' X}  
\begin{pmatrix}
e^{i \arccos(x)} & 0 \\
0 & e^{-i \arccos(x)}
\end{pmatrix},
\end{equation}
we have 
\begin{equation}
\mc{U}\ket{0}\ket{v_j}=(U_{\Phi_z}(\cos\lambda_j)\ket{0}) \ket{v_j}=F(\cos \lambda_j) \ket{0}\ket{v_j} + \alpha_j\ket{1}\ket{v_j}.
\end{equation}
Here $x=\cos\lambda_j$, and $\alpha_j=Q(\cos\lambda_j)\sin\lambda_j+i G(\cos\lambda_j)$.
The rest of the proof follows that of \cref{thm:qet_unitary}.
\end{proof}

\section{Cost of \QETU using Trotter formulas}
\label{sec:trotter}

If the time evolution operator $U=e^{-iH}$ is implemented using Trotter formulas, we can directly analyze the circuit depth and gate complexity of estimating the ground-state energy in the setting of Theorems \ref{thm:ground_state_energy_woaa} and \ref{thm:ground_state_energy}.

We suppose the Hamiltonian $H$ can be decomposed into a sum of terms $H=\sum_{\gamma=1}^L H_\gamma$, where each term $H_{\gamma}$ can be efficiently exponentiated, i.e. with gate complexity independent of time. In other words we assume each $H_{\gamma}$ can be fast-forwarded \cite{AtiaAharonov2017,GuSommaSahinoglu2021fast,Su2021fast}. We assume that the gate complexity for implementing a single Trotter step is $G_{\mathrm{Trotter}}$, and the circuit depth required is $D_{\mathrm{Trotter}}$. For initial state preparation, we assume we need gate complexity $G_{\mathrm{initial}}$ and circuit depth $D_{\mathrm{initial}}$. 
A $p$-th order Trotter formula applied to $U=e^{-i H}$ with $r$ Trotter steps gives us a unitary operator $U_{\mathrm{HS}}$ with error
\[
\|U_{\mathrm{HS}}-U\|\leq C_{\mathrm{Trotter}} r^{-p},
\]
where $C_{\mathrm{Trotter}}$ is a prefactor, for which the simplest bound is $C_{\mathrm{Trotter}}=\Or((\sum_{\gamma}\|H\|_{\gamma})^{p+1})$. Tighter bounds in the form of a sum of commutators are proved in Refs.~\cite{ChildsSuTranEtAl2021,SuHuangCampbell2020nearly}, and there are many works on how to decompose the Hamiltonian to reduce the resource requirement \cite{McArdleCampbellSu2022exploiting,LeeBerryEtAl2021even,BerryGidneyEtAl2019qubitization,vonBurgLowEtAl2021quantum}.
If the circuit queries $U,U^{\dag}$ for $d$ times and the desired precision is $\delta$, then we can choose
\[
d\times C_{\mathrm{Trotter}} r^{-p}=\delta,
\]
or equivalently
\begin{equation}\label{eqn:qetu_trotter_r}
r= \Or\left(\max\{d^{1/p} C_{\mathrm{Trotter}}^{1/p}\delta^{-1/p},1\}\right).
\end{equation}

As an example, let us now analyze the number of Trotter steps needed in the context of estimating the ground-state energy in Theorem \ref{thm:ground_state_energy_woaa} without amplitude amplification. When replacing the exact $U$ with $U_{\mathrm{HS}}$, we only need to ensure that the resulting error in $U_{\mathrm{proj}}$ defined in \cref{eq:defn_U_proj} is  $\Or(\gamma)$. This will enable us to solve the binary amplitude estimation problem (see Definition \ref{defn:binary_amplitude_estimation}) with the same asymptotic query complexity. Since $U_{\mathrm{proj}}$ uses $U$ (and therefore $U_{\mathrm{HS}}$) at most $\wt{\Or}(\epsilon^{-1}\log(\gamma^{-1}))$ times, we only need to ensure
\begin{equation}
\label{eq:trotter_err_condition}
    \wt{\Or}(\epsilon^{-1}\log(\gamma^{-1})) \times C_{\mathrm{Trotter}} r^{-p} = \Or(\gamma).    
\end{equation}
Consequently we need to choose
\begin{equation}
\label{eq:number_trotter_steps}
    r= \wt{\Or}\left(\max\{C_{\mathrm{Trotter}}^{1/p}\epsilon^{-1/p}\gamma^{-1/p},1\}\right).
\end{equation}
The query depth of $U$ is $\wt{\Or}(\epsilon^{-1}\log(\gamma^{-1}))$, and therefore the circuit depth is
\[
\wt{\Or}\left(\max\{C_{\mathrm{Trotter}}^{1/p}\epsilon^{-1/p}\gamma^{-1/p},\log(\gamma^{-1})\}\epsilon^{-1}D_{\mathrm{Trotter}}+D_{\mathrm{initial}}\right).
\]
Similarly the total number of queries to $U$ is $\wt{\Or}(\epsilon^{-1}\gamma^{-2}\log(\vartheta^{-1}))$ times, and the total number of queries to $U_I$ is $\Or(\gamma^{-2}\polylog(\epsilon^{-1}\vartheta^{-1}))$, and the resulting total gate complexity is
\[
\wt{\Or}\left(\max\{C_{\mathrm{Trotter}}^{1/p}\epsilon^{-1/p}\gamma^{-1/p},1\}\epsilon^{-1}\gamma^{-2}G_{\mathrm{Trotter}}\log(\vartheta^{-1})+\gamma^{-2}G_{\mathrm{initial}}\log(\vartheta^{-1})\right).
\]

The analysis of the number of Trotter steps in the setting of Theorem \ref{thm:ground_state_energy} is similar. We still want to ensure \eqref{eq:trotter_err_condition}, which results in the same choice of the number of Trotter steps $r$ as in \eqref{eq:number_trotter_steps}. Combined with the query complexity in Theorem \ref{thm:ground_state_energy}, the total gate complexity is
\begin{equation}
    \label{eq:near_optimal_gate_complexity}
    \wt{\Or}\left(\max\{C_{\mathrm{Trotter}}^{1/p}\epsilon^{-1/p}\gamma^{-1/p},1\}\epsilon^{-1}\gamma^{-1}G_{\mathrm{Trotter}}\log(\vartheta^{-1})+\gamma^{-1}G_{\mathrm{initial}}\log(\vartheta^{-1})\right).
\end{equation}

\section{Binary amplitude estimation with a single ancilla qubit and \QETU}
\label{sec:binary_amplitude_estimation}

In this section we discuss how to solve the binary amplitude estimation problem in Definition \ref{defn:binary_amplitude_estimation} using a single ancilla qubit. We restate the problem here.
 
\begin{defn*}[Binary amplitude estimation]
Let $W$ be a unitary acting on two registers, with the first register indicating success or failure. Let $A=\|(\bra{0}\otimes I_n)W(\ket{0}\ket{0^n})\|$ be the success amplitude. Given $0\leq\gamma_1<\gamma_2$, provided that $A$ is either smaller than $\gamma_1$ or greater than $\gamma_2$, we want to correctly distinguish between the two cases, i.e. output 0 for the former and 1 for the latter. 
\end{defn*}
In the following we will describe an algorithm to solve this problem and thereby prove Lemma \ref{lem:binary_amplitude_estimation}.
We can find quantum states $\ket{\Phi}$ and $\ket{\perp}$ such that
\[
W(\ket{0}\ket{0^n}) = A\ket{0}\ket{\Phi} + \sqrt{1-A^2}\ket{\perp},
\]
and $(\bra{0}\otimes I)\ket{\perp}=0$. We also define
\[
\ket{\perp'} = -\sqrt{1-A^2}\ket{0}\ket{\Phi} + A\ket{\perp}.
\]

As in amplitude amplification, we define two reflection operators: 
\[
R_0 = (2\ket{0}\bra{0}-I)\otimes I_n,\quad R_1=W(2\ket{0^{n+1}}\bra{0^{n+1}}-I_{n+1})W^{\dagger}.
\]

Relative to the basis $\{W(\ket{0}\ket{0^n}),\ket{\perp'}\}$ the two reflection operators can be represented by the matrices
\[
\begin{pmatrix}
2A^2-1 & -2A\sqrt{1-A^2} \\
-2A\sqrt{1-A^2} & 1-2A^2
\end{pmatrix},
\quad
\begin{pmatrix}
1 & 0 \\
0 & -1
\end{pmatrix}.
\]
Therefore we can verify that $\ket{\Psi_{\pm}}=(W\ket{0}\ket{0^n}\pm i\ket{\perp'})/\sqrt{2}$ are eigenvectors of $R_0 R_1$:
\[
R_0 R_1 \ket{\Psi_{\pm}} = e^{\mp i 2\arccos(A)}  \ket{\Psi_{\pm}}.
\]

If we use the usual amplitude estimation algorithm to estimate $A$ we can simply perform phase estimation with $R_0 R_1$ on the quantum state $W(\ket{0}\ket{0^n})$, which is an equal superposition of $\ket{\Psi_{\pm}}$:
\[
W(\ket{0}\ket{0^n}) = \frac{1}{\sqrt{2}}(\ket{\Psi_{+}}+\ket{\Psi_{-}}).
\]
However, here we will do something different. We view $R_0 R_1$ has a time-evolution operator corresponding to some Hamiltonian $L$:
\[
R_0 R_1 = e^{-iL},
\]
where, in the subspace spanned by $\{W(\ket{0}\ket{0^n}),\ket{\perp'}\}$ we have
\[
L = 2\arccos(A)\ket{\Psi_+}\bra{\Psi_+} - 2\arccos(A)\ket{\Psi_-}\bra{\Psi_-}.
\]
Then, using \QETU in Theorem \ref{thm:qet_unitary} we can implement a block encoding, which we denote by $\mathcal{U}$, of $P(\cos(L/2))$ for any suitable polynomial $P$, and in the same subspace we have
\[
P(\cos(L/2)) = P(A)\big(\ket{\Psi_+}\bra{\Psi_+}+\ket{\Psi_-}\bra{\Psi_-}\big) = P(A)\big(W\ket{0}\ket{0^n}\bra{0}\bra{0^n}W^{\dagger}+\ket{\perp'}\bra{\perp'}\big).
\]
Using \QETU we can use Monte Carlo sampling to estimate the quantity
\begin{equation}
\label{eq:prob_getting_0_binary_amplitude_estimation}
    \|(\bra{0}\otimes I_{n+1})\mathcal{U}\big(\ket{0}\otimes (W\ket{0}\ket{0^n})\big)\|^2 = |P(A)|^2.
\end{equation}
To be more precise, we can start from the state $\ket{0}\ket{0}\ket{0^n}$ on $n+2$ qubits, apply $W$ to the last $n+1$ qubits, then $\mathcal{U}$ to all $n+2$ qubits, and in the end measure the first qubit. The probability of obtaining $0$ in the measurement outcome is exactly \cref{eq:prob_getting_0_binary_amplitude_estimation}.

Now let us consider an even polynomial $P(x)$ such that $|P(x)|\leq 1$ for $x\in[-1,1]$, and
\[
P(x)\geq 1-\delta,\ x\in[\gamma_2,1],\quad |P(x)|\leq \delta,\ x\in[0,\gamma_1].
\]
Such a polynomial of degree $\Or((\gamma_2-\gamma_1)^{-1}\log(\delta^{-1}))$ can be constructed using the approximate sign function in \cite{LowChuang2017a} (if we take this approach we need to symmetrize the polynomial through $P(x)=(Q(x)+Q(-x))/2$) or the optimization procedure described in \cref{sec:convex}. Using this polynomial, we can then use Monte Carlo sampling to distinguish two cases, which will solve the binary amplitude estimation problem:
\[
\|(\bra{0}\otimes I_{n+1})\mathcal{U}\big(\ket{0}\otimes (W\ket{0}\ket{0^n})\big)\|^2 \geq (1-\delta)^2,\ \text{or}\ \|(\bra{0}\otimes I_{n+1})\mathcal{U}\big(\ket{0}\otimes (W\ket{0}\ket{0^n})\big)\|^2\leq \delta^2.
\]
We can choose $\delta=1/4$ and it takes running $W$ and $\mathcal{U}$ and measuring the first qubit each $\Or(\log(\vartheta^{-1}))$ times to successfully distinguish between the above two cases with probability at least $1-\vartheta$. In this we use the standard majority voting procedure to boost the success probability. Each single run of $\mc{U}$ requires $\Or((\gamma_2-\gamma_1)^{-1})$ applications of $W$, which corresponds to the polynomial degree. Therefore in total we need to apply $W$ $\Or((\gamma_2-\gamma_1)^{-1}\log(\vartheta^{-1}))$ times.

In this whole procedure we need one additional ancilla qubit for \QETU. 
Note that the $n+1$ qubits reflection operator in $R_1$ can be implemented using the $(n+2)$-qubit Toffoli gate and phase kickback. Using \cite[Corollary 7.4]{BarencoBennettEtAl1995elementary} we can implement the $(n+2)$-qubit Toffoli gate on $(n+3)$ qubits. As a result another ancilla qubit is needed. We have proved Lemma \ref{lem:binary_amplitude_estimation}, which we restate here:  
\begin{lem*}
The binary amplitude estimation problem in Definition~\ref{defn:binary_amplitude_estimation} can be solved correctly with probability at least $1-\vartheta$ by querying $W$ $\Or((\gamma_2-\gamma_1)^{-1}\log(\vartheta^{-1}))$ times, and this procedure requires two additional ancilla qubits (besides the ancilla qubits already required in $W$).
\end{lem*}

\section{Details of numerical simulation of TFIM}\label{sec:num-detail}

In the numerical simulation for estimating the ground-state energy of TFIM, we explicitly diagonalize the Hamiltonian to obtain the exact ground state $\ket{\psi_0}$, ground energy $E_0$, the first excited energy $E_1$, and the highest excited energy $E_{n-1}$. We then perform an affine transformation to the shifted Hamiltonian
\begin{equation}
    H^\mathrm{sh} = c_1 H + c_2 I_n, \quad c_1 = \frac{\pi-2\eta}{E_{n-1} - E_0} \text{ and } c_2 = \eta - c_1 E_0.
\end{equation}
Consequently, the eigenvalues of the shifted Hamiltonian are exactly in the interval $[\eta, \pi - \eta]$, i.e., $E_0^\mathrm{sh} = \eta$ and $E_{n-1}^\mathrm{sh} = \pi - \eta$. The time evolution is then $e^{-\I \tau H^\mathrm{sh}} = e^{-\I \tau c_2} e^{-\I \tau c_1 H}$ which means that the evolution time is scaled to $\tau^\mathrm{sh} = \tau c_1$ with an additional phase shift $\phi^\mathrm{sh} = \tau c_2$. The system dependent parameters are then given by
\begin{equation}
    \mu = \frac{1}{2}\left(E_0^\mathrm{sh} + E_1^\mathrm{sh}\right),\ \Delta = E_1^\mathrm{sh} - E_0^\mathrm{sh}, \text{ and } \sigma_\pm = \cos\frac{\mu \mp \Delta/2}{2}.
\end{equation}
We set the input quantum state to $\ket{0} \ket{\psi_\mathrm{in}}$ where $\ket{\psi_\mathrm{in}} = \ket{0^n}$ and the additional one qubit is the ancilla qubit for performing $X$ rotations in \QETU. The initial overlap is $\gamma = \abs{\braket{\psi_\mathrm{in} | \psi_0}}$. We list the system dependent parameters used in the numerical experiments in \cref{fig:ground_energy} in \cref{tab:system_params}.

\begin{table}[htbp]
\centering
\begin{tabular}{@{} *{8}{c} @{}}\hline
$n$ & $\mu$ & $\Delta$ & $\sigma_+$ & $\sigma_-$ & $c_1$ & $c_2$ & $\gamma$\\\hline
2 & 0.7442 & 1.2884 & 0.9988 & 0.7686 & 0.1824 & 1.5708 & 0.5301 \\
4 & 0.3926 & 0.5851 & 0.9988 & 0.9419 & 0.0909 & 1.5708 & 0.3003 \\
6 & 0.2887 & 0.3773 & 0.9988 & 0.9717 & 0.0605 & 1.5708 & 0.1703 \\
8 & 0.2394 & 0.2788 & 0.9988 & 0.9821 & 0.0453 & 1.5708 & 0.0965 \\
\hline
\end{tabular}
\caption{System dependent parameters for different number of system qubits $n$.}
\label{tab:system_params}
\end{table}

For completeness, we briefly introduce the algorithm for deriving the energy estimation from the measurement of bit-string frequencies. The energy estimation process can be optimized so that it is sufficient to measure a few quantum circuits to compute the energy. For \TFIM, if the ground state is $\ket{\psi_0}$, its ground-state energy is 
\begin{equation*}
    E_0 = \braket{\psi_0 | H_\TFIM | \psi_0} = - \sum_{j=1}^{n-1} \braket{\psi_0 | Z_j Z_{j+1} | \psi_0} - g \sum_{j=1}^n \braket{\psi_0 | X_j | \psi_0} =: - \sum_{j=1}^{n-1} \Psi_{j,j+1} - g \sum_{j=1}^n \Psi_j^H.
\end{equation*}
We will show that the energy component $\Psi_{j,j+1}$ and $\Psi_j^H$ can be exactly expressed as the marginal probabilities readable from measurements. Decomposing $Z_j Z_{j+1}$ with respect to eigenvectors, we have
\begin{equation*}
    \Psi_{j,j+1} = \braket{\psi_0 | Z_jZ_{j+1} | \psi_0} = \sum_{z_j = 0}^1 \sum_{z_{j+1} = 0}^1 (-1)^{z_j+z_{j+1}} \abs{\braket{\psi_0 | z_j,z_{j+1}}}^2 = \sum_{z_j = 0}^1 \sum_{z_{j+1} = 0}^1 (-1)^{z_j+z_{j+1}} \bP\left(z_j, z_{j+1}\right | \psi_0).
\end{equation*}
Here, $\bP\left(z_j, z_{j+1} | \psi_0\right)$ is the marginal probability measuring the $j$-th qubit with $z_j$ and the $(j+1)$-th qubit with $z_{j+1}$ under computational basis when the quantum circuit for preparing the ground state $\ket{\psi_0}$ is given. Similarly, the other quantity involved in the energy is
\begin{equation*}
    \Psi^H_j = \braket{\psi_0 | X_j | \psi_0} = \braket{\psi_0 |H^{\otimes n} Z_j H^{\otimes n}| \psi_0} = \braket{\psi_0^H | Z_j | \psi_0^H} = \sum_{z_j = 0}^1 (-1)^{z_j} \bP\left(z_j | \psi_0^H\right).
\end{equation*}
Here, $\bP\left(z_j | \psi_0^H\right)$ is the marginal probability measuring the $j$-th qubit with $z_j$ under computational basis when the quantum circuit for preparing the ground state $\ket{\psi_0}$ following a Hadamard transformation, which is denoted as $\ket{\psi_0^H} := H^{\otimes n}\ket{\psi_0}$, is given.

In order to estimate the ground-state energy of \TFIM, it suffices to measure all qubits in two circuits: the circuit in \cref{fig:TFIM_circuits} (b) and that following a Hadamard transformation on all system qubits. The measurement results estimate the marginal probabilities up to the Monte Carlo measurement error. Furthermore, their linear combination with signs gives the ground-state energy estimate based on the previous analysis. 

The procedure for estimating the energy can readily be generalized to other models. Consider a Hamiltonian
\begin{equation}
    H = \sum_{k=1}^L H_k, \quad H_k = \sum_{j=1}^{v_k} h_{k,j}.
\end{equation}
Here we  group the components of the Hamiltonian into $L$ classes, and for a fixed $k$, the components $h_{k,j}$ can be simultaneously diagonalized by an efficiently implementable unitary $V_k$. \REV{The strategies of Hamiltonian grouping has also been used in e.g., Refs. \cite{IzmaylovYenLangEtAl2019,VerteletskyiYenIzmaylov2020}.} We want to estimate the expectation $\braket{\psi_0 | H |\psi_0}$ where $\ket{\psi_0}$ is the quantum state prepared by some quantum circuit. Then, it suffices to measure $L$ different quantum circuits $\{ V_k \ket{\psi_0} : k = 1, \cdots, L \}$ and to compute the expectation from the measurement data by some signed linear combination. For example, to estimate the ground-state energy of the Heisenberg model, we can let $L = 3$ and $V_1 = I^{\otimes n}$, $V_2 = \mathrm{H}^{\otimes n}$ and $V_3 = \left(\mathrm{H}\mathrm{S}^\dagger\right)^{\otimes n}$ where $\mathrm{H}$ and $\mathrm{S}$ are Hadamard gate and phase gate respectively.

\bibliographystyle{abbrvnat}
\bibliography{ref,lin_ref,heisenberg}

\end{document}